\newif\iffull

\fullfalse

\documentclass[11pt]{article}

\usepackage{hyperref}

\usepackage{setspace}
\usepackage[\iffull notes=true, \else notes=false, \fi later=false]{dtrt}
\usepackage{fullpage}
\usepackage{microtype}
\usepackage{amssymb,amsfonts,amsmath,amsthm}
\usepackage{fullpage}
\usepackage{graphicx}
\usepackage{verbatim}
\usepackage{array}
\usepackage{latexsym}
\usepackage{paralist}
\usepackage{enumitem}
  \setlist[itemize]{leftmargin=*}
  \setlist[enumerate]{leftmargin=*}
\usepackage[capitalise,nameinlink]{cleveref}
\usepackage{dsfont}
\usepackage{url}
\usepackage{braket}
\usepackage{mathrsfs}
\usepackage{color}
\usepackage{complexity}
\usepackage{multirow}
\usepackage{makecell}
\usepackage{mdframed}
\usepackage[margin=1cm,small]{caption}
\usepackage[normalem]{ulem}
\usepackage[vlined,linesnumbered,ruled,resetcount]{algorithm2e}
\usepackage{algpseudocode}
\crefname{algocf}{Algorithm}{Algorithms}
\Crefname{algocf}{Algorithm}{Algorithms}
\SetKwInOut{InputExplicity}{Exlpicit input}
\SetKwInOut{InputOracle}{Oracle access}
\SetKwInOut{Output}{output}

\newtheorem{stheorem}{Theorem}
\newtheorem{hypothesis}{Hypothesis}
\newtheorem{corollary}{Corollary}[section]
\newtheorem{theorem}[corollary]{Theorem}

\newtheorem{definition}[corollary]{Definition}
\newtheorem{lemma}[corollary]{Lemma}
\newtheorem{claim}[corollary]{Claim}
\newtheorem{observation}[corollary]{Observation}

\newtheorem{question}[corollary]{Question}

\newtheorem*{lemma*}{Lemma}
\theoremstyle{definition}

\newtheorem{remark}[corollary]{Remark}

\definecolor{forestgreen}{rgb}{0.13, 0.55, 0.13}

\allowdisplaybreaks
\newcommand{\FormatAuthor}[3]{
\begin{tabular}{c}
#1 \\ {\small\texttt{#2}} \\ {\small #3}
\end{tabular}
}
\newcommand{\doclearpage}{
\iffull
\clearpage
\else
\fi
}
\newcommand{\defemph}[1]{\textbf{\emph{\textit{#1}}}}
\newcommand{\keywords}[1]{\bigskip\par\noindent{\footnotesize\textbf{Keywords\/}: #1}}

\newcommand{\colored}[1]{\ifthenelse{\boolean{color}}{\textcolor{red}{#1}}{#1}}
\newcommand{\green}[1]{\ifthenelse{\boolean{green}}{\textcolor{green}{#1}}{#1}}

\newcommand{\DefineEqual}{\colored{:=}}

\newcommand{\Alphabet}{\Sigma}
\newcommand{\Bits}{\colored{\{0,1\}}}
\newcommand{\Cardinality}[1]{\colored{\left|#1\right|}}
\newcommand{\Absolute}[1]{\colored{\left|#1\right|}}
\newcommand\abs[1]{\Absolute{#1}}

\newcommand{\SoundnessError}{\gamma}

\newcommand{\Event}{E}

\newcommand{\NSStrategy}{{\mathcal{F}}}
\newcommand{\NSStrategyRep}[1][\NumRep]{\NSStrategy^{(#1)}}
\newcommand{\CorrectNSStrategyRep}[1][\NumRep]{\Correct{{{\NSStrategy}^{(#1)}}}}

\newcommand{\MaxQueries}{{k}}

\newcommand{\QuasiLin}{\mathcal{L}}
\newcommand{\QuasiLinRep}[1][\NumRep]{\colored{\mathcal{L}^{(#1)}}}
\newcommand{\QuasiConsistent}[1][\NumRep]{\colored{\mathcal{C}^{(#1)}}}

\newcommand{\NumRep}{t}

\newcommand{\hyp}{{\mathsf{hyp}}}
\newcommand{\expHypothesis}{{\mathsf{e_{\hyp}}}}

\newcommand{\Distance}[1]{\colored{\Delta_{#1}}}

\newcommand{\floor}[1]{\colored{{\lfloor#1\rfloor}}}

\newcommand{\N}{{\mathbb N}}

\newcommand{\eps}{\varepsilon}
\newcommand{\seq}{\subseteq}

\newcommand{\ip}[1]{\langle #1 \rangle}

\newcommand{\AddFunction}{{\mathsf{add}}}

\newcommand{\nsPCP}{\mathsf{ns}\PCP}

\newcommand{\nsLPCP}{\mathsf{ns}\LPCP}

\newclass{\LPCP}{LPCP}

\newcommand{\Correct}[1]{{\widehat{#1}}}
\newcommand{\Fold}[1]{{\overline{#1}}}

\newcommand{\Flat}[1]{{\widetilde{#1}}}

\newcommand{\ConstraintPoly}{P}
\newcommand{\ConstraintConstant}{c}
\newcommand{\NumConstraints}{M}

\newcommand{\Verifier}{\mathsf{V}}

\newcommand{\DecisionPredicate}{\mathsf{D}}
\newcommand{\DecisionPredicateLIN}{\DecisionPredicate_{\LIN}}

\newcommand{\Language}{L}

\newcommand{\InputSize}{n}

\newcommand{\Domain}{D}
\newcommand{\Answers}{\Sigma}

\newcommand{\Diagonal}[1]{D_{#1}}

\newcommand{\Circuit}{C}
\newcommand{\Wire}{w}

\newcommand{\NumWires}{N}
\newcommand{\Gate}{g}
\newcommand{\PCPParams}[8]{{#8 \left[
{\small
\begin{array}{r  l}
\textsf{soundness error:}&\enspace {#1} \\
\textsf{randomness:}&\enspace {#2} \\
\textsf{proof length:}&\enspace {#3} \\
\textsf{query complexity:}&\enspace {#4} \\
\textsf{locality:}&\enspace {#5} \\
\textsf{query sampler time:}&\enspace {#6}\\
\textsf{decision predicate time:}&\enspace {#7}\\
\end{array}
}
\right]}}

\newcommand{\Soundness}{\eps}

\newcommand{\linlocalityhat}{\green{\bar{\MaxQueries}}}
\newcommand{\linlength}{{\green{n}}}

\newcommand{\WiresVars}{\mathbf{\Wire}}
\newcommand{\nsLPCPParams}[7]{\PCPParams{#1}{#2}{#3}{#4}{#5}{#6}{#7}{\nsLPCP}}
\newcommand{\nsPCPParams}[7]{\PCPParams{#1}{#2}{#3}{#4}{#5}{#6}{#7}{\nsPCP}}

\newcommand{\flatten}[1]{{\textsf{Flat}}[ #1 ]}
\newboolean{color}
\setboolean{color}{false}
\newboolean{green}
\setboolean{green}{false}

\begin{document}

\title{Toward Probabilistic Checking against Non-Signaling Strategies with Constant Locality}

\author{
\begin{tabular}[h!]{ccc}
\FormatAuthor{Mohammad Mahdi Jahanara}{mjahanar@sfu.ca}{Simon Fraser University} &
  \FormatAuthor{Sajin Koroth}{sajin\_koroth@sfu.ca}{Simon Fraser University}
     & \FormatAuthor{Igor Shinkar}{ishinkar@sfu.ca}{Simon Fraser University}
\end{tabular}
} %

\date{\today}

\maketitle

\begin{abstract}

Non-signaling strategies are a generalization of quantum strategies that have been studied in physics over the past three decades. Recently, they have found applications in theoretical computer science, including to proving inapproximability results for linear programming and to constructing protocols for delegating computation. A central tool for these applications is probabilistically checkable proof (PCPs) systems that are \emph{sound against non-signaling strategies}.

In this paper we show, assuming a certain geometrical hypothesis about noise robustness of non-signaling proofs (or, equivalently, about robustness to noise of solutions to the Sherali-Adams linear program), that a slight variant of the parallel repetition of the exponential-length constant-query PCP construction due to Arora et al.\ (JACM 1998) is sound against non-signaling strategies with \emph{constant locality}.

Our proof relies on the analysis of the \emph{linearity test} and \emph{agreement test} (also known as the \emph{direct product test}) in the non-signaling setting.

\end{abstract}
\newpage
\tableofcontents
\newpage

\section{Introduction}
\label{sec:introduction}
Probabilistically Checkable Proofs (PCPs)~\cite{BabaiFLS91,FeigeGLSS96,AroraS98,AroraLMSS98} are proofs that can be verified by a probabilistic verifier that queries only a few locations of the proof. PCPs have been a powerful tool in the theory of computing, with applications in diverse areas such as hardness of approximation~\cite{FeigeGLSS96} and delegation of computation~\cite{Kilian92,Micali00}.
A seminal result of~\cite{AroraS98,AroraLMSS98}, known as the PCP theorem, says that every language
decidable by a non-deterministic Turing machine in time $T(n)$
has a PCP system which allows to check if a given input of length $n$ is in the language
by using $O(\log(T(n)))$ random bits and making only $O(1)$ queries to the given proof.

Recall that in the classical setting of PCPs the two standard requirements are completeness and soundness.
\emph{Completeness} requires that if a given input is in the language, then there is some proof that convinces the prover with probability $1$.
\emph{Soundness} requirement states that if the input is not in the language, the prover rejects \emph{any} proof with some significant probability.
In this paper we study PCP systems that are sound against \emph{non-signalling proofs} or \emph{non-signalling strategies},
i.e., we require the prover to reject \emph{any} non-signalling proof with some significant probability.

Non-signalling strategies are a certain restricted class of probabilistic oracles.
When such oracle is given a set of queries, the response to the queries is sampled from a distribution such that the answer to each query may depend on all queries.
More precisely, a non signalling strategy with locality $\MaxQueries$ is a collection $\NSStrategy = \set{\NSStrategy_S}_{S\subseteq \Domain,|S|\leq \MaxQueries}$,
where each $\NSStrategy_S$ is a distribution over $\Alphabet^S$ (i.e., over functions $f:S \to \Alphabet$),
and for any two subsets $S,T \subseteq \Domain$ of size at most $\MaxQueries$,
the restrictions of $\NSStrategy_S$ and $\NSStrategy_T$ to $S\cap T$ are equal as distributions. This setting stands in contrast to the standard notion of a classical proof, where the answer to each query is deterministic.
Note that if the locality is the maximum possible, i.e., $\MaxQueries = |\Domain|$, then $\NSStrategy$ is a distribution over functions,
which is (essentially) equivalent to the classical notion of a proof.

We note that one may think about $\MaxQueries$-non-signalling functions, equivalently,
as the class of all feasible solutions to the linear program arising from the
$\MaxQueries$'th level relaxation of the Sherali-Adams hierarchy~\cite{SheraliA90}.
This implies that computing the maximum acceptance probability of an nsPCP verifier
that uses $r$ random bits,
where the maximum is taken over all $\MaxQueries$-non-signaling proofs,
reduces to a linear program with $2^{O(r \cdot \MaxQueries^2)} \cdot \Alphabet^{O(\MaxQueries)}$ variables and constraints.
In particular, if a language $\Language$ has a PCP verifier
that on an input of length $n$ uses $r = O(\log(n))$ random bits,
and is sound against $O(1)$-non-signaling proofs over an alphabet of constant size,
then $\Language$ is decidable in time $\poly(n)$.

Non-signaling strategies have been studied in physics since 1980's~\cite{Rastall85, KhalfinT85, PopescuR94}
in order to better understand quantum entanglement.
Indeed, these strategies strictly generalize quantum strategies
and capture minimal requirements on ``non-local'' correlations
that rule out instantaneous communication.

PCP systems that are sound against \emph{non-signalling proofs} have recently found
numerous applications in theoretical computer science, including
schemes for 1-round delegation of computation from cryptographic assumptions \cite{KalaiRR13, KalaiRR14},
and hardness of approximation for linear programming \cite{KalaiRR16}.
However, as opposed to the well studied setting of the classical PCP theorem,
where there are many constructions achieving best parameters possible,
in the non-signalling setting many parameters of the known PCP constructions appear to be far from optimal.

One of the most important parameters associated with a non-signalling proof is the locality parameter, denoted by $\MaxQueries$.
Indeed, \cite{KalaiRR13,KalaiRR14} have studied the related notion of multi-prover interactive proofs that are sound against non-signaling
strategies (nsMIPs). They have shown that nsPCPs are essentially equivalent to nsMIPs
where $\MaxQueries$, the locality of the proof in the nsPCP setting,
exactly corresponds to the number of provers in the nsMIP setting.

Despite the importance of the locality parameter, the exact complexity of languages admitting nsPCPs
that are sound against $\MaxQueries$-non-signaling proofs is still open for most $\MaxQueries$'s.
Note that as the locality of the proof decreases, there are fewer constraints imposed on the proof,
and hence the task of the verifier becomes more challenging.
The seminal result of Kalai, Raz, and Rothblum \cite{KalaiRR13,KalaiRR14}
showed that every language in $\DTIME(T)$ has an nsPCP verifier that uses $\polylog(T)$
random bits, makes $\polylog(T)$ queries to a proof of length $poly(T)$,
and is sound against $\polylog(T)$ -non-signaling proofs.
In particular, every language in $\EXP$ is captured by a nsMIP
with a polynomial time randomized verifier who communicates with $\poly(n)$ non-signaling provers.
For the limitations of nsPCPs, Ito~\cite{Ito10} proved that for $\MaxQueries=2$ the corresponding linear program is solvable in $\PSPACE$,
which is tight by the result of~\cite{ItoKM09}, and hence the class $\PSPACE$ is captured by PCPs that are sounds against $2$-non-signaling proofs.
Much less is known about the power of PCP systems that are sound against $\MaxQueries$-non-signaling proofs for $\MaxQueries > 2$.
Recently, Holden and Kalai~\cite{HoldenKalali20} proved that $o(\sqrt{\log(n)})$-prover non-signalling proofs with \emph{negligible soundness} is contained in $\PSPACE$.

All these results give rise to the following question, raised in \cite{ChiesaMS19},
asking for the non-signaling analogue of the PCP theorem.

\begin{question}
\label{qst:ns-pcp-thm}
Is it true that every language in $\DTIME(T)$
has an nsPCP verifier that uses $O(\log(T))$ random bits, makes $O(1)$ queries to the proof, and is sound against $O(1)$-non-signalling functions?
\end{question}

Motivated by this problem, Chiesa et.~al~\cite{ChiesaMS19,ChiesaMS20} started a systematic study of non-signalling PCPs.
They proposed studying the classical (algebraic) PCP constructions and their building blocks
(which are very well understood in the classical setting), and adapting each of the building blocks to the non-signaling setting.
In particular, focusing on the PCP construction of \cite{AroraLMSS98}
they made an appropriate definition of \emph{linear} non-signalling functions and analyzed the linearity test of \cite{BlumLR93}
against non-signalling strategies~\cite{ChiesaMS20}.
Then, building on the linearity test, they proved in~\cite{ChiesaMS19} that the classical
exponential length $O(1)$-query PCP of~\cite{AroraLMSS98} is sound against $O(\log^2(N))$-non-signalling proofs.
We emphasize, that even for exponential length nsPCPs (corresponding to nsPCPs with $r = \poly(n)$ randomness),
there are no known constructions that are sound against $O(1)$-non-signaling proofs.
Given this state of affairs, it is natural to ask the following question, that is simpler than \cref{qst:ns-pcp-thm}

\begin{question}
\label{qst:ns-pcp-thm-exp}
Is it true that every language in $\DTIME(T(n))$
has an nsPCP verifier that uses $O(\poly(T(n))$ random bits, makes $O(1)$ queries to the proof, and is sound against $O(1)$-non-signalling functions?
\end{question}

One must be careful with the precise formulation of \cref{qst:ns-pcp-thm-exp}.
Note that if the verifier uses more than $T(n)$ random bits, the runtime spent on reading the randomness
is more than $T(n)$, which is the time complexity of the problem.
To recover a nontrivial question, we require the verifier to be \emph{input oblivious}.
That is, in order to decide whether an instance $x$ belongs to the given language $\Language \in \DTIME(T(n))$,
the verifier generates the queries based only on the length of the input $x$ and its randomness (but not the input itself),
and then rules according to an $o(T)$-time decision predicate (where the predicate does depend on $x$).
Indeed, the \cite{AroraLMSS98} verifier studied in~\cite{ChiesaMS19} is input oblivious.

In this work we build on the work of~\cite{ChiesaMS19} and provide a positive answer to \cref{qst:ns-pcp-thm-exp}
\emph{assuming a certain geometric hypothesis}.
Specifically, we construct an input oblivious nsPCP verifier for any language $\Language \in \DTIME(T(n))$
that uses $\poly(T(n))$ random bits, makes $O(1)$ queries to a given proof, and is sound against $O(1)$-non-signalling functions,
with \emph{two caveats}.
\begin{enumerate}
  \item The first is that the alphabet of the nsPCP system is $\Alphabet = \Bits^{\polylog(T(n))}$,
  instead of the binary alphabet employed by \cite{ChiesaMS19,KalaiRR14,AroraLMSS98}.
  Still, this means that the verifier reads a total of $\polylog(T(n))$ bits from the proof, which makes our result non-trivial.
  Also, recall that in the classical setting, we have the alphabet reduction technique using proof composition,
  and it is plausible that we can apply similar ideas also in the non-signaling setting.
  Indeed, proof composition is an important building block in the classical PCP literature,
  and we believe it will also be an important step toward resolving \cref{qst:ns-pcp-thm}.

  \item The second caveat is that our result depends on a certain quantitative geometric hypothesis
  about proximity between almost non-signaling proofs and exactly non-signaling proofs.
  Equivalently, the hypothesis says that every feasible solution for the noisy Sherali-Adams LP
  is close (in some precise, rather weak, sense) to a feasible solution for the (exact) Sherali-Adams LP.
  See \cref{hyp:rounding} for details, and the discussion in \cref{sec:hypothesis-discussion}.
\end{enumerate}

Our work follows the general philosophy of \cite{ChiesaMS19, ChiesaMS20}, who proposed building modular analogues
of tools and techniques from the classical PCP literature.
A classical tool used in the construction of PCPs is parallel repetition~\cite{Raz98, Holenstein09}.
In the classical setting of 2-query PCP, parallel repetition is used to reduce the soundness error.
In this work we use parallel repetition for non-signalling proofs to reduce the locality to $O(1)$,
while the soundness stays in the ``high-probability acceptance regime''.
In addition to parallel repetition, we study additional tools from the PCP literature.
Specifically, we use the modular approach that is typical for the classical setting.
Specifically, we show first that  the parallel repetition of the \cite{AroraLMSS98} verifier is sound against
``nicely structured'' proofs.
Then, we use \emph{linearity test} and \emph{direct product test}, and claim that
proofs that satisfy both tests with high probability must be nicely structured,
and hence we essentially reduce the analysis to the structured case.

Another interesting feature of our proof is the reduction from the parallel repetition
of the~\cite{AroraLMSS98} verifier to the non-repeated \cite{AroraLMSS98} verifier.
Specifically, we show that if for some input $x$,
the parallel repetition of the~\cite{AroraLMSS98} verifier
accepts a proof with high probability, and the proof is ``nicely structured'',
then it is possible to ``flatten'' the repeated proof into a proof over the binary alphabet,
that satisfies the (non-repeated) \cite{AroraLMSS98} verifier with high probability.
Therefore, by applying the result of~\cite{ChiesaMS19} about the soundness of the \cite{AroraLMSS98} verifier,
we conclude that the input $x$ is in the language.

\subsection{Informal statement of the result}
\label{sec:main-result-informal}

Below we discuss the main result of the paper.
Our result depends on an hypothesis about approximating \emph{almost non-signaling} functions
using \emph{exactly non-signaling} functions.

\begin{hypothesis}[Informal]
\label{hyp:rounding-informal}
    Any almost linear, almost non-signaling function $\NSStrategy \colon \Bits^n \to \Bits$
    can be well approximated by some non-signaling function $\NSStrategy' \colon \Bits^n \to \Bits$ of slightly lower locality.

    Equivalently, any solution to the noisy Sherali-Adams LP
    can be well approximated by a solution to the (exact) Sherali-Adams LP of slightly lower level in the hierarchy.
\end{hypothesis}

The exact formulation of the hypothesis relies on the precise definitions
of \emph{non-signaling} and \emph{almost non-signaling} functions (or, equivalently, the related notions of \emph{noisy Sherali-Adams LP}),
as well as the appropriate definitions of distance.
For the formal statement of the hypothesis see \cref{hyp:rounding} following the required definitions in \cref{sec:prelims}.

We are now ready to state our main theorem.

\begin{stheorem}[Main theorem - informal]
\label{thm:main-informal}
Assuming \cref{hyp:rounding-informal}
every language $\Language \in \DTIME(T)$ has an input oblivious nsPCP verifier that
on an input of length $n$ uses $\widetilde{O}(T^2)$ random bits,
makes $O(1)$ queries to proofs over the alphabet $\Sigma=\Bits^{\polylog(T)}$, and is sound against $O(1)$-non-signaling proofs.
The query sampler runs in time $\widetilde{O}(T^2)$, and the decision predicate runs in time $O(\InputSize \cdot \polylog(T))$.
\end{stheorem}

To the best our knowledge, this is the first result that constructs a PCP system that is sound against non-signaling proofs with constant locality.

\subsection{Roadmap}
\label{sec:roadmap}

The rest of the paper is organized as follows.
In \cref{sec:prelims} we formally define the notions that we utilize throughout this work, and use them to formally state our hypothesis and the main theorem in \cref{sec:main-result}. In \cref{sec:PCP-construction} we recall the ALMSS verifier, and define our variant of its parallel repetition.
In \cref{sec:overview} we provide an overview of the soundness proof.
In \cref{sec:linear-proof} we prove soundness of our verifier against structured proofs. In \cref{sec:test-self-correction} we discuss our local testing and self-correction, which enables us to reduce soundness against general proofs to soundness against structured proofs. Finally, in \cref{sec:proof-of-soundness} we prove the main result.

\doclearpage
\section{Preliminaries}
\label{sec:prelims}

\subsection{Probabilistically Checkable Proofs}
We start with the definition of Probabilistically Checkable Proofs (PCPs).
Recall that a classical PCP verifier for a language $\Language$ is given an input $x$, and an oracle access to a proof.
The verifier reads the input, uses randomness, queries the proof in a small number of coordinates,
and based on the answers to the queries decides whether to accept or reject.
Completeness requires that if $x \in \Language$, then there exists a proof that makes the verifier always accept.
Soundness requires that if $x \not\in \Language$, then for any proof the verifier will reject with high probability.

In the non-signaling setting, a \emph{non-signaling PCP verifier} is a verifier, whose soundness is further required to hold against any non-signaling proof of prescribed locality.
More precisely, an nsPCP verifier $\Verifier$ for a language $\Language$ gets an input $x$ and an oracle access to a non-signaling function $\NSStrategy \colon \Domain \to \Answers$.
The verifier reads the input $x$, uses random bits to decide on a subset $S \subseteq \Domain$ on which $\NSStrategy$ is queried.
Then, based on the answer $\NSStrategy(S) \in \Answers^S$ it decides to accept or reject.

\begin{definition}
\label{def:ns-pcp}
A \defemph{nsPCP verifier} for a language $\Language \seq \Bits^*$ is a randomized algorithm $\Verifier$ that
gets an input $x \in \Bits^n$ and oracle access to a $\MaxQueries$-non-signaling proof $\NSStrategy \colon \Domain \to \Alphabet$.
The verifier uses randomness to decide on a subset $S \subseteq \Domain$ of size $\abs{S} \leq \MaxQueries$,
and queries $\NSStrategy$ on $S$.
Then, based on the answer $\NSStrategy(S) \in \Answers^S$ it decides to accept or reject.
We say that $\Verifier$ has perfect completeness and soundness error $\SoundnessError$ against $\MaxQueries$-non-signaling proofs
if the following holds.
\begin{description}
  \item[Completeness:]\label{item:pcp-completeness}
  For all $x \in \Language$ there exists a (classical) proof $\pi$ such that $\Pr[\Verifier_{\pi}(x) = 1] = 1$.
  \item[Soundness:]\label{item:pcp-soundness}
  If $x \notin \Language$, then for all $\MaxQueries$-non-signaling proofs $\NSStrategy$ it holds that $\Pr[\Verifier_{\NSStrategy}(x) = 1] \leq \SoundnessError$.
\end{description}
We say that verifier $\Verifier$ is \defemph{input oblivious} if the choice of the query set $S$ depends only on
the input length $n$, the randomness of the verifier, but is independent of $x$.
\end{definition}

\begin{remark}
Note that in the non-signaling setting the locality parameter $\MaxQueries$ upper bounds the number of queries made by the verifier,
and it is possible that the actual predicate used by the verifier depends on significantly less than $\MaxQueries$ coordinates of the proof.
For example, \cite{ChiesaMS19} proved that the 11-queries verifier of \cite{AroraLMSS98} is sound against $O(\log^2(n))$-non-signaling proofs,
and it is not known whether the verifier is sound against $O(1)$-non-signaling proofs, or even $o(\log^2(n))$-non-signaling proofs.
\end{remark}

\subsection{Parallel repetition}
In the classical setting a proof is assumed to be a string, or equivalently, a static function $\pi: \Domain \to \Alphabet$ committed by the prover. A $\NumRep$-parallel repetition of a proof $\pi$ is a mapping $\pi^\NumRep \colon \Domain^\NumRep \to \Alphabet^\NumRep$
that allows accessing $\NumRep$ locations of the (supposed) proof by making only 1 query to a (longer) proof over a larger alphabet.
That is, the intended proof $\pi^{(\NumRep)}$ corresponds to some ``base'' proof
$\pi \colon \Domain \to \Sigma$ defined as $\pi^{(\NumRep)}((x_1,\dots,x_\NumRep)) = (\pi(x_1),\dots,\pi(x_\NumRep))$.
Analogously, given a verifier $\Verifier$, a $t$-repeated verifier which is denoted by $\Verifier^{(\NumRep)}$,
runs $\NumRep$ parallel \textit{independent} instances of $\Verifier$ and accepts if and only if all instances accept.

The original motivation for using parallel repetition was to reduce the soundness error of a proof system,
while keeping the number of queries fixed. In the classical setting, if the repeated proof is indeed
a parallel repetition of some base proof $\pi$, then it is not hard to see that the soundness error of $\Verifier^{(t)}_{\pi^{t}}$
is exponentially smaller than the soundness error of $\Verifier_{\pi}$.
The soundness analysis of the repeated proof need not be based on this comparison to the soundness error of the \textit{base} proof,
and analyzing such proofs in both classical and non-signalling settings has been a subject of a long line of research~\cite{Verbitsky96, Raz98, Holenstein09, DinurS14, BravermanGarg2015, LancienW16, HolmgrenY19}.

In this work, we use parallel repetition to improve the \emph{minimum locality parameter} of non-signaling proofs required for the soundness of the verifier, rather than its soundness error. Next, we formally define non-signaling proofs, and some properties of such proofs that we will need in the paper.

\subsection{Non-signaling functions}

In this work we consider PCP verifiers that are sound against non-signaling proofs.
Below, we formally define the notion of non-signaling functions, and introduce some notation we will use in the paper.
Throughout the paper we will use terms \emph{non-signaling function}, \emph{non-signaling proof}, and \emph{non-signaling strategy} interchangeably.

\begin{definition}
\label{def:ns-function}
Fix a domain $\Domain$, an alphabet $\Sigma$, and a parameter $\MaxQueries \in \N$.
A \defemph{$\MaxQueries$-non-signaling function} $\NSStrategy \colon \Domain \to \Sigma$ is a collection $\NSStrategy = \{\NSStrategy_{S}\}_{S \subseteq \Domain, \Cardinality{S} \leq \MaxQueries}$, where
each $\NSStrategy_S$ is a distribution over assignments $f_S \colon S \to \Sigma$, such that
for every two subsets $S, T \seq \Domain$ each of size at most $\MaxQueries$, the marginal distributions of $\NSStrategy_{S}$ and $\NSStrategy_{T}$ restricted to $S \cap T$ are equal.
\end{definition}

Unlike a classic function, we can use a $\MaxQueries$-non-signaling function only once
in the sense that one has to present the set of at most $\MaxQueries$ queries all at once.
In other words, it is not possible to use the non-signaling function adaptively.

\begin{remark}
  Throughout the paper we will consider non-signaling functions of two types:
  \begin{itemize}
  \itemsep0em
    \item functions over the domain $\Domain = \Bits^N$ for some $N \in \N$ and alphabet $\Alphabet=\Bits$;
    \item functions over the domain $\Domain = (\Bits^N)^\NumRep$ and alphabet $\Alphabet=\Bits^{\NumRep}$ for some parameters $N,\NumRep \in \N$.
  \end{itemize}
\end{remark}

Next, we define a relaxed notion of non-signaling functions, that allows the marginal distributions induced by different query sets to be only \textit{statistically close} rather equal on the intersection. This relaxation arises in our analysis. It has also appeared naturally in other works in this area, especially in cryptographic applications~\cite{AielloBOR00, DworkLNNR04, KalaiRR13, KalaiRR14}.
\begin{definition}
\label{def:almost-ns-function}
Fix a domain $\Domain$, an alphabet $\Sigma$, and parameters $\MaxQueries \in \N$ and $\eps \in [0,1]$.
A \defemph{$(\eps, \MaxQueries)$-non-signaling function} over a domain $\Domain$ and an alphabet $\Sigma$, is a collection $\NSStrategy = \{\NSStrategy_{S}\}_{S \subseteq \Domain, \Cardinality{S} \leq \MaxQueries}$, where
each $\NSStrategy_S$ is a distribution over assignments $f_S \colon S \to \Sigma$, such that
for every two subsets $S, T \seq \Domain$ each of size at most $\MaxQueries$, the marginal distributions of $\NSStrategy_{S}$ and $\NSStrategy_{T}$ restricted to $S \cap T$ are $\eps$-close with respect to total variation distance,
i.e.,
\begin{equation*}
    \max_{E \seq \Alphabet^{S \cap T}}
    \left|\Pr_{\NSStrategy_S}[\NSStrategy_S {\mid_{S \cap T}} \in E] - \Pr_{\NSStrategy_T}[\NSStrategy_{T} {\mid_{S \cap T}} \in E] \right|
    \leq \eps
    \enspace.
\end{equation*}
In particular, a $(\eps=0, \MaxQueries)$-non-signaling-function coincides with the
definition of $\MaxQueries$-non-signaling function from \cref{def:ns-function}.
\end{definition}

\noindent Next we define non-signaling and almost non-signaling counterpart of parallel repeated functions.

\begin{definition}
\label{def:repeated-ns-function}
Fix a domain $\Domain$, an alphabet $\Sigma$, and parameters $\MaxQueries, \NumRep \in \N$.
A \defemph{$\NumRep$-repeated $(\delta,\MaxQueries)$-non-signaling function} is an $(\delta,\MaxQueries)$-non-signaling function $\NSStrategyRep \colon\Domain^\NumRep \to \Sigma^\NumRep$.
Namely, a $\NumRep$-repeated $(\delta,\MaxQueries)$-non-signaling function $\NSStrategyRep \colon \Domain^\NumRep \to \Alphabet^\NumRep$ is a collection $\NSStrategyRep = \{\NSStrategyRep_{S}\}_{S \subseteq \Domain^\NumRep, \Cardinality{S} \leq \MaxQueries}$, where
each $\NSStrategyRep_S$ is a distribution over assignments $f^{(\NumRep)}_S \colon S \to \Sigma$, such that
for every two subsets $S, T \seq \Domain^\NumRep$ each of size at most $\MaxQueries$, the marginal distributions of $\NSStrategyRep_{S}$ and $\NSStrategyRep_{T}$ restricted to $S \cap T$ are $\delta$-close with respect to total variation distance.
\end{definition}

We will also need the definition of distance between non-signaling or almost non-signaling functions.

\begin{definition}[Statistical distance]
\label{def:distance}
    Let $\NSStrategy, \NSStrategy' \colon \Domain \to \Alphabet$ be two non-signaling or almost non-signaling functions
    with locality $\MaxQueries$.
    For $\ell \leq \MaxQueries$ the $\Distance{\ell}$-distance between $\NSStrategy$ and $\NSStrategy'$ is defined as
    \begin{equation*}
        \Distance{\ell}(\NSStrategy, \NSStrategy') = \max_{S \seq \Domain, \abs{S} \leq \ell} \Delta(\NSStrategy_S, \NSStrategy'_S)
        \enspace,
    \end{equation*}
    where $\Delta(\NSStrategy_S, \NSStrategy'_S) = \max_{E \seq \Alphabet^S} \abs{\Pr[\NSStrategy_S \in E] - \Pr[\NSStrategy'_S \in E]}$
    is the total variation distance between $\NSStrategy_S$ and $\NSStrategy'_S$.

    We say that $\NSStrategy$ and $\NSStrategy'$ are $\eps$-close in the $\Distance{\ell}$-distance if
    $\Distance{\ell}(\NSStrategy, \NSStrategy') \leq \eps$, and say that they are $\eps$-far otherwise.
\end{definition}

\subsection{Permutation folded repeated non-signaling functions}

Folding is a technique used to impose some structure on the given proof without really making extra queries.
The idea of using folded proofs was first introduced by \cite{BellareGS98}.
We formally define the \emph{permutation folding} property,
and then explain why we can impose this property without making extra queries.

\begin{definition}
\label{def:permutation-folding}
Let $Q = (q_1, \dots, q_\NumRep) \in \Domain^\NumRep$ be a $\Domain$-values vector, and let $\pi \in S_\NumRep$ be a permutation of the indices $[\NumRep]$. Define $\pi(Q) = (q_{\pi(1)}, \dots, q_{\pi(\NumRep)})$ to be the vector obtained from $Q$ by permuting the coordinates according to $\pi$.

Let $\NSStrategyRep \colon (\Domain^n)^\NumRep \to \Alphabet^\NumRep$ be a $\NumRep$-repeated $\MaxQueries$-non-signaling function.
$\NSStrategyRep$ is said to be \defemph{permutation folded} or \defemph{permutation invariant} if
for any $S = \{Q_1, \dots, Q_\ell\} \seq (\Domain^n)^\NumRep$ with $1 \leq \ell \leq \MaxQueries$,
for any $T = \{\pi_1(Q_1), \dots, \pi_\ell(Q_\ell) \}$
for some permutations $\pi_1,\dots\pi_\ell \in S_\NumRep$,
and for any $b_1, \dots, b_\ell \in \Alphabet^\NumRep$ it holds that
\begin{equation*}
    \Pr \left[ \forall i \in [\ell] \quad \NSStrategyRep_S(Q_i) = b_i \right]
         =
    \Pr \left[ \forall i \in [\ell] \quad \NSStrategyRep_T(\pi_i(Q_i)) = \pi_i(b_i)
         \right]
         \enspace.
\end{equation*}
\end{definition}

\begin{observation}\label{obs:permutation-folding}
It is important to note that we can \emph{fold} any given $\NumRep$-repeated $\MaxQueries$-non-signaling function $\NSStrategyRep \colon \Domain^\NumRep \to \Alphabet^\NumRep$ by partitioning $\Domain^\NumRep$ into equivalence classes, where $Q$ and $Q'$ belong to the same class if $Q' = \pi(Q)$ for some permutation $\pi$.

We defined the folding of $\NSStrategyRep$, denoted by $\Fold\NSStrategyRep$ as follows. For any query $Q$ to $\Fold\NSStrategyRep$ ,
let $\pi \in S_t$ be a uniformly random permutation, and define the distribution of
$\Fold\NSStrategyRep(Q)$ as the distribution of $\pi^{-1}(\NSStrategyRep(\pi(Q)))$.

It is easy to see that $\Fold\NSStrategyRep$ is indeed $\MaxQueries$-non-signaling and permutation folded.
Furthermore, note that if $\NSStrategyRep$ is permutation folded, then $\Fold\NSStrategyRep = \NSStrategyRep$.
\end{observation}

\subsection{Linear non-signaling functions}

In this part, we define \emph{linear} $\NumRep$-repeated non-signaling functions.
Linear non-signaling \emph{boolean} functions have been studied in \cite{ChiesaMS20, ChiesaMS19},
and played a key role in the proving that the PCP verifier of \cite{AroraLMSS98} is sound against non-signaling proofs.
We also use such structured non-signaling proofs in this paper. See \cref{sec:PCP-construction} for details.

\begin{definition}[Linear $\NumRep$-repeated functions]
\label{def:linear-rep-ns-fun}
Let $\QuasiLinRep \colon (\Bits^n)^\NumRep \to \Bits^\NumRep$ be a $\NumRep$-repeated $(\eps,\MaxQueries)$-non-signaling function.
We say that $\QuasiLinRep$ is \emph{linear} if for all $X,Y \in (\Bits^n)^\NumRep$,
and $X+Y \in (\Bits^n)^\NumRep$ defined by the coordinate-wise addition modulo 2,
and for all $S \seq (\Bits^n)^\NumRep$ containing $X,Y,X+Y$
of size at most $\abs{S} \leq \MaxQueries$,
it holds that
\begin{equation*}
    \Pr_{\QuasiLinRep_{S}}\left[\QuasiLinRep(X) + \QuasiLinRep(Y) = \QuasiLinRep(X+Y) \right] = 1
    \enspace.
\end{equation*}
\end{definition}

\begin{remark}
    Note that in the degenerate case of $\NumRep=1$
    if a (non-repeated) $\MaxQueries$-non-signaling function $\NSStrategy$ satisfies the linearity condition in \cref{def:linear-rep-ns-fun}
    then $\Pr\left[\NSStrategy(x) + \NSStrategy(y) = \NSStrategy(x+y) \right] = 1$ for all $x,y \in \Bits^n$,
    i.e., $\NSStrategy$ satisfies the linearity test of \cite{BlumLR93} with probability 1.
    Non-signaling functions satisfying this property have been the subject of work on linearity testing in the non-signaling setting~\cite{ChiesaMS20}.
\end{remark}

Next we extend \cref{def:linear-rep-ns-fun} by introducing the notion of an \emph{almost} linear $\NumRep$-repeated non-signalling function.

\begin{definition}[Almost linear $\NumRep$-repeated functions]
\label{def:almost-linear-rep-ns-fun}
Let $\QuasiLinRep \colon (\Bits^n)^\NumRep \to \Bits^\NumRep$ be a $\NumRep$-repeated $(\delta,\MaxQueries)$-non-signaling function.
We say that $\QuasiLinRep$ is $(1-\eps)$-\emph{linear} if for all $X,Y \in (\Bits^n)^\NumRep$,
and $X+Y \in (\Bits^n)^\NumRep$ defined by the coordinate-wise addition modulo 2,
and for all $S \seq (\Bits^n)^\NumRep$ containing $X,Y,X+Y$
of size at most $\abs{S} \leq \MaxQueries$,
it holds that
\begin{equation*}
    \Pr_{\QuasiLinRep_{S }}\left[\QuasiLinRep(X) + \QuasiLinRep(Y) = \QuasiLinRep(X+Y) \right] \geq 1-\eps
    \enspace.
\end{equation*}
\end{definition}

We will allow ourselves to use the informal term \emph{almost linear},
when referring to a non-signaling function $\QuasiLinRep$ that is $(1-\eps)$-linear for some small $\eps$.

\subsection{Consistent repeated non-signaling functions}

In this part, we define the notion of \emph{consistency} for $\NumRep$-repeated $\MaxQueries$-non-signaling function.

\begin{definition}[Consistent $\NumRep$-repeated functions]
\label{def:consistent-rep-ns-fun}
Let $\QuasiConsistent \colon \Domain^\NumRep \to \Alphabet^\NumRep$ be a $\NumRep$-repeated $\MaxQueries$-non-signaling function.
We say that $\QuasiConsistent$ is \emph{consistent}, if for any $Q,Q' \in \Domain^\NumRep$ it holds that
\begin{equation*}
    \Pr_{\QuasiConsistent}\left[\QuasiConsistent(Q)_j = \QuasiConsistent(Q')_j \quad \forall j \in [\NumRep] \mbox{ such that } Q_j=Q'_j \right] = 1
    \enspace.
\end{equation*}

\end{definition}

Similarly to the almost linear property, we define the relaxed notion of almost consistent non-signalling function.

\begin{definition}[Almost consistent $\NumRep$-repeated functions]
\label{def:almost-consistent-rep-ns-fun}
Let $\QuasiConsistent \colon \Domain^\NumRep \to \Alphabet^\NumRep$ be a $\NumRep$-repeated $\MaxQueries$-non-signaling function.
We say that $\QuasiConsistent$ is $(1-\eps)$-\emph{consistent}, if for any $Q, Q' \in \Domain^\NumRep$
\begin{equation*}
    \Pr_{\QuasiConsistent}\left[\QuasiConsistent(Q)_j = \QuasiConsistent(Q')_j \quad \forall j \in [\NumRep] \mbox{ such that } Q_j=Q'_j \right] \geq 1 - \eps
    \enspace.
\end{equation*}
\end{definition}

We will allow ourselves to use the informal term \emph{almost consistent},
when referring to a non-signaling function $\QuasiConsistent$ that is $(1-\eps)$-consistent for some small $\eps$.

\begin{claim}\label{claim:consistent-F}
Let $\QuasiConsistent \colon \Domain^\NumRep \to \Alphabet^\NumRep$ be a $\NumRep$-repeated
$\MaxQueries$-non-signaling function for $\MaxQueries \geq 3$,
and suppose that $\QuasiConsistent$ is $(1-\eps)$-\emph{consistent}.
Fix $Q,Q' \in \Domain^\NumRep$ and let $J = \{ j\in [\NumRep] : Q_j = Q'_j\}$.
Then, for any event $E \seq \Alphabet^J$ it holds that
\begin{equation*}
    \abs{ \Pr[\QuasiConsistent(Q)_{J} \in E] - \Pr[\QuasiConsistent(Q')_{J} \in E] } \leq \eps
    \enspace.
\end{equation*}

\end{claim}
\begin{proof}
    Note that
    \begin{align*}
      \Pr[\QuasiConsistent(Q)_{\mid J} \in E]
      & \geq \Pr[\QuasiConsistent(Q)_{\mid J} \in E \wedge \QuasiConsistent(Q)_{\mid J} = \QuasiConsistent(Q')_{\mid J}] \\
      & = \Pr[\QuasiConsistent(Q')_{\mid J} \in E \wedge \QuasiConsistent(Q)_{\mid J} = \QuasiConsistent(Q')_{\mid J}] \\
      & \geq \Pr[\QuasiConsistent(Q')_{\mid J} \in E]  - \eps
      \enspace,
    \end{align*}
    where the last inequality is by the assumption that $\QuasiConsistent$ is $(1-\eps)$-consistent.
    By symmetry, we also get the inequality in the other direction, and the claim follows.
\end{proof}

We observe that for $\Domain = \Bits^n$ and $\Alphabet = \Bits$ (almost) linearity implies (almost) consistency.
Specifically, we prove the following claim.

\begin{claim}
\label{claim:lin-implies-consistent}
  Let $\QuasiLinRep \colon (\Bits^n)^\NumRep \to \Bits^\NumRep$ be a $\NumRep$-repeated $\MaxQueries$-non-signaling function,
  and suppose that
  \begin{inparaenum}[(i)]
  \item $\QuasiLinRep$ is $(1-\eps)$-linear, and
  \item  $\Pr\left[ \QuasiLinRep(Q)_j = 0 \quad \forall j \in [\NumRep] \mbox{ such that } Q_j = 0^n \right] > 1 - \eps$
  for all $Q \in (\Bits^n)^\NumRep$.
  \end{inparaenum}
  Then, $\QuasiLinRep$ is $(1-2\eps)$-consistent when treated as a $(\MaxQueries-1)$-non-signaling function.
\end{claim}

\begin{proof}
Let $S \in (\Bits^n)^\NumRep$ be a set of queries of size $\abs{S} \leq \MaxQueries-1$.
Let $Q,Q' \in S$, and let $J = \{j \in [\NumRep] : Q_j = Q'_j$.
We show below that
\begin{equation*}
    \Pr_{\QuasiConsistent_{Q,Q'}}
    \left[\QuasiConsistent(Q)_j = \QuasiConsistent(Q')_j \quad \forall j \in J \right] \geq 1 - \eps
    \enspace.
\end{equation*}
Consider the set of queries $S' = S \cup \{Q''\}$, where $Q'' = Q+Q'$.
In particular, $Q''_j = 0^n$ for all $j \in J$.
By the assumption of the claim we get that
$\Pr\left[\QuasiLinRep(Q'')_j = 0 \quad \forall j \in J \right] \geq 1-\eps$.
Therefore, using the assumption that $\QuasiLinRep$ is $(1-\eps)$-linear it follows that
\begin{align*}
    \Pr\left[\QuasiConsistent(Q)_j \neq \QuasiConsistent(Q')_j \quad \forall j \in J \right]
    &\geq \Pr\left[\QuasiLinRep(Q)_j + \QuasiLinRep(Q')_j = \QuasiLinRep(Q'')_j \wedge \QuasiLinRep(Q'')_j = 0 \quad \forall j \in J \right]\\
    &\geq 1 - 2\eps
    \enspace.
\end{align*}
Therefore, $\QuasiLinRep$ is $(1-2\eps)$-consistent, as required.
\end{proof}

\subsection{Flattening of a \texorpdfstring{$\NumRep$}{t}-nsPCP}
\label{subsec:flatteningOfnsPCP}
Below we define the \emph{flattening} operation, which transforms a given $\NumRep$-repeated proof into a non-repeated proof in the natural way.
Namely, given a query set $S$ to the non-repeated proof, we create a vector $Q^S$ containing all the elements of $S$,
query the repeated proof on $Q^{S}$, and respond according to the received answer.

\begin{definition}
\label{def:flatten}
    Let $\NSStrategyRep \colon \Domain^\NumRep \to \Alphabet^\NumRep$ be a $\MaxQueries$-non-signaling $\NumRep$-repeated proof.
    Define the flattening of $\NSStrategyRep$, denoted by $\Flat\NSStrategy = \flatten{\NSStrategyRep} \colon \Domain \to \Alphabet$ as follows.
    For a query set $S = \{q_1,\dots,q_s\} \seq \Domain$ of size $s \leq \NumRep$,
    define a vector $Q^S$ whose first $s$ entries are $(q_1,\dots,q_s)$ and the rest are set arbitrarily,
    query $\NSStrategyRep$ on the single query $Q^S$, and let the distribution of
    $\Flat\NSStrategy(S)$ be
    \begin{equation*}
      \Flat\NSStrategy(S) = (\NSStrategyRep(Q^S)_1, \dots, \NSStrategyRep(Q^S)_s)
      \enspace.
    \end{equation*}
\end{definition}

\begin{claim}
\label{claim:flatten-almost-consistent}
Let $\QuasiConsistent \colon \Domain^\NumRep \to \Alphabet^\NumRep$  be a $\MaxQueries$-non-signaling function that is permutation folded and $(1-\eps)$-consistent for $k \geq 2$.
Then $\Flat\NSStrategy = \flatten{\QuasiConsistent}$ is a $(\eps, \NumRep)$-non-signaling function.

Furthermore, fix a query $Q =(w_1,\dots,w_\NumRep) \in \Domain^\NumRep$ for $\QuasiConsistent$,
a query set $S \seq \Domain$ of size $s$ for $\NSStrategy$,
also let $1 \leq \ell \leq \NumRep$ such that $w_1, \dots, w_\ell$ are distinct and $w_j \in S$ for all $j \in [\ell]$.
Then, the distribution of $\NSStrategy_S(\{w_1,\dots,w_\ell\})$
and
$(\QuasiConsistent(Q)_1, \dots, \QuasiConsistent(Q)_\ell)$ are $\eps$-close in total variation distance.
\end{claim}
\begin{proof}
    To prove that $\Flat\NSStrategy$ is $(\eps, \NumRep)$-non-signaling function
    let $S, T \in \Domain$ be two sets of queries, and suppose $S \cap T = \{w_1, \dots, w_\ell\}$.
    We want to show that for any event $E \seq \Alphabet^{S \cap T}$ it holds that
    \begin{equation}\label{eq:flatten-almost-consistent-key-ineq}
        \left|
        \Pr_{\Flat\NSStrategy_S}[\Flat\NSStrategy_S {\mid_{S \cap T}} \in E]
        - \Pr_{\Flat\NSStrategy_T}[\Flat\NSStrategy_{T} {\mid_{S \cap T}} \in E]
        \right| \leq \eps
        \enspace.
    \end{equation}
    Define $Q^S, Q^T \in \Domain^{\NumRep}$ as in \cref{def:flatten}, let $\pi, \pi' \in S_\NumRep$ be permutations such that for all $j \in [\ell]$ it holds that $\pi(Q^S)_j = \pi'(Q^T)_j = w_j$. By non-signaling and permutation invariance of $\QuasiConsistent$, if we query it on $\{\pi(Q^S), \pi'(Q^T)\}$ we have:
    \begin{align*}
        \Pr_{\Flat\NSStrategy_S}[\Flat\NSStrategy_S {\mid_{S \cap T}} \in E] &=
        \Pr\left[ \left(\QuasiConsistent(\pi_1(Q^S)), \dots, \QuasiConsistent(\pi(Q^S))_\ell \right) \in E \right]
        \\
        \Pr_{\Flat\NSStrategy_T}[\Flat\NSStrategy_T{\mid_{S \cap T}} \in E] &=
        \Pr\left[ \left(\QuasiConsistent(\pi'(Q^T))_1, \dots, \QuasiConsistent(\pi'(Q^T))_\ell \right) \in E \right] \enspace.
    \end{align*}
    Then, by \cref{claim:consistent-F} we get the following:
    \[\abs{\Pr\left[ \left(\QuasiConsistent(\pi_1(Q^S)), \dots, \QuasiConsistent(\pi(Q^S))_\ell \right) \in E \right] - \Pr\left[ \left(\QuasiConsistent(\pi'(Q^T))_1, \dots, \QuasiConsistent(\pi'(Q^T))_\ell \right) \in E \right]} \leq \eps\]
    which proves \cref{eq:flatten-almost-consistent-key-ineq}.
    Therefore, $\Flat\NSStrategy$ is a $(\eps, \NumRep)$-non-signaling function.

    \medskip

    Next we prove the second part of the claim.
    Given $S$, define $Q^S \in \Domain^{\NumRep}$ as in \cref{def:flatten},
    and consider the query set $\{Q^S,Q\}$ to $\QuasiConsistent$.
    Since $\QuasiConsistent$ is permutation folded, we may assume that
    $Q_j = Q^S_j = w_j$ for all $j \in [\ell]$.
    Therefore, for any $E \seq \Alphabet^{\ell}$ we have:
    \begin{align*}
      &\abs{\Pr\left[\left(\Flat\NSStrategy_S(w_1), \dots, \Flat\NSStrategy_S(w_\ell)\right) \in E\right] - \Pr\left[\left(\QuasiConsistent(Q)_1, \dots, \QuasiConsistent(Q)_\ell\right) \in E\right]} \\
      &=\abs{\Pr\left[\left(\QuasiConsistent(Q^S)_1, \dots, \QuasiConsistent(Q^S)_\ell\right) \in E\right] - \Pr\left[\left(\QuasiConsistent(Q)_1, \dots, \QuasiConsistent(Q)_\ell\right) \in E\right]}
      \enspace,
    \end{align*}
    which is upper bounded by $\eps$ by \cref{claim:consistent-F}.
    This complets the proof of \cref{claim:flatten-almost-consistent}
\end{proof}

The following claim is follows rather immediately from \cref{claim:flatten-almost-consistent} above.

\begin{claim}
\label{claim:flatten-linear-proof}
Let $\MaxQueries \geq 4$, and let
$\QuasiLinRep \colon (\Bits^n)^\NumRep \to \Bits^\NumRep$ be a
$\MaxQueries$-non-signaling function that is permutation folded, $(1-\eps_1)$-linear,
and $(1-\eps_2)$-consistent.
Then $\Flat{\QuasiLin} = \flatten{\QuasiLinRep}$ is a (non-repeated) $(\eps_2,\NumRep)$-non-signaling $(1-\eps_1-3\eps_2)$-linear function.
\end{claim}

\begin{proof}
    By applying \cref{claim:flatten-almost-consistent} on $\QuasiLinRep$,
    we get that $\Flat{\QuasiLin} = \flatten{\QuasiLinRep}$ is a $(\eps_2, \NumRep)$-non-signaling function.
    Next we prove that $\Flat{\QuasiLin}$ is $(1-(\eps_1 + 3\eps_2))$-linear.
    Fix $x,y \in \Bits^n$,
    and let $S \seq \Bits^n$ be a query set for $\Flat{\QuasiLin}$ such that $\{x,y,x+y\} \seq S$.
    We want to prove that
    \begin{equation}\label{eq:flatten-linear-proof-key-ineq}
        \Pr[\Flat{\QuasiLin}(x) + \Flat{\QuasiLin}(y) = \Flat{\QuasiLin}(x+y)] \geq 1 - \eps_1 - 3\eps_2
        \enspace.
    \end{equation}
    Let $Q^S$ be as in \cref{def:flatten}.
    By the permutation folding property of $\QuasiLinRep$ we may assume that the first three coordinates of $Q^S$ are $x,y,x+y$.
    That is $Q^S_1 = x, Q^S_2 = y$, and $Q^S_3 = x+y$.

    By definition of $Q^S$ we have
    $\Pr[\Flat{\QuasiLin}(x) + \Flat{\QuasiLin}(y) = \Flat{\QuasiLin}(x+y)] = \Pr[\QuasiLinRep(Q^S)_1 + \QuasiLinRep(Q^S)_2 = \QuasiLinRep(Q^S)_3]$.
    Consider now the vectors $Q^x = (x,0^n,0^n,\dots,0^n)$, $Q^y = (y,0^n,0^n,\dots,0^n)$, and $Q^{x+y} = (x+y,0^n,0^n,\dots,0^n)$.
    Since $\QuasiLinRep$ is $(1-\eps_1)$-linear, we get that
    $\Pr[\QuasiLinRep(Q^x) + \QuasiLinRep(Q^y) = \QuasiLinRep(Q^{x+y})] \geq 1-\eps_1$.
    Since $\QuasiLinRep$ is $(1-\eps_2)$-consistent,
    it follows that
    \begin{align*}
        \Pr[\Flat{\QuasiLin}(x) + \Flat{\QuasiLin}(y) = \Flat{\QuasiLin}(x+y)]
        & = \Pr[\QuasiLinRep(Q^S)_1 + \QuasiLinRep(Q^S)_2 = \QuasiLinRep(Q^S)_3] \\
        & \geq \Pr[\QuasiLinRep(Q^x) + \QuasiLinRep(Q^y) = \QuasiLinRep(Q^{x+y})] - 3\eps_2 \\
        & \geq 1 - \eps_1 - 3\eps_2
        \enspace,
    \end{align*}
    as required.
\end{proof}

\doclearpage
\section{Main result}
\label{sec:main-result}

In this section we formally state the main result of the paper.
In order to describe the result we need to first state the hypothesis conditioned on which our main theorem holds.

\begin{hypothesis}\label{hyp:rounding}
    Fix integers $n$ and $k \leq 2^n$, and let $\eps \in (0,1)$.
    For any $(\eps, \MaxQueries)$-almost non-signaling function $\NSStrategy \colon \Bits^n \to \Bits$
    that is $(1-\eps)$-linear there exists a $\MaxQueries'$-non-signaling function $\NSStrategy' \colon \Bits^n \to \Bits$ such that
    $\Distance{4}(\NSStrategy, \NSStrategy') \leq \eps'$,
    where $\MaxQueries' \geq \MaxQueries^{\expHypothesis}$ for some positive absolute constant $\expHypothesis > 0$,
    and $\eps'=\eps'_{\hyp}(\eps)$ is some function that depends only on $\eps$ such that $\eps'_{\hyp}(\eps) \to 0$ as $\eps \to 0$.
\end{hypothesis}

\begin{remark}
We make two remarks regarding the hypothesis.
\begin{itemize}
    \item A statement analogous to \cref{hyp:rounding} has been proven in \cite{ChiesaMS20},
        showing that there exist a $\MaxQueries$-non-signaling function $\NSStrategy' \colon \Bits^n \to \Bits$
        such that $\Distance{\MaxQueries}(\NSStrategy, \NSStrategy') \leq O(4^\MaxQueries \cdot \eps)$.
        The multiplicative factor of $4^\MaxQueries$ is too large, which makes it insufficient for our applications.
    \item For our applications, we need a much weaker version of \cref{hyp:rounding}.
        We elaborate more on the hypothesis in \cref{sec:hypothesis-discussion}.
\end{itemize}
\end{remark}

For a computable function $\NumWires \colon \N \to \N$
we denote by $\SIZE(\NumWires)$ the complexity class of all languages $\Language$
having a uniform family of boolean circuits $(\Circuit_n \colon \Bits^n \to \Bits)_{n \in \N}$
of maximum fan-in 2 with AND, OR, and NOT gates,
such that $\Circuit_n$ has at most $\NumWires(n)$ wires for all $n \in \N$.%
\footnote{Note that our complexity measure for the size of a circuit is the number of wires, (and not the number of gates, which is more standard)
as this measure directly affects the complexity of the PCP construction.
However, for circuits with bounded fan-in, the two quantities are equal up to a multiplicative constant factor.}

\begin{stheorem}[Main theorem]
\label{thm:main}
Assuming \cref{hyp:rounding}
every language $\Language \in \SIZE(\NumWires)$ has an input oblivious nsPCP verifier that
on input of length $n$ uses $\widetilde{O}(\NumWires^2)$ random bits,
makes $O(1)$ queries to proofs over the alphabet $\Sigma=\Bits^{\polylog(\NumWires)}$, and is sound against $O(1)$-non-signaling proofs.
The query sampler runs in time $\widetilde{O}(\NumWires^2)$, and the decision predicate runs in time $O(\InputSize \cdot \polylog(\NumWires))$.
That is,
\begin{equation*}
\SIZE(\NumWires)
\subseteq
\nsPCPParams{1-\Omega(1)}{\widetilde{O}(\NumWires^2)}{2^{\widetilde{O}(\NumWires^2)}}{4}{O(1)}{\widetilde{O}(\NumWires^2)}{O(\InputSize \cdot \polylog(\NumWires))}
\enspace.
\end{equation*}
\end{stheorem}

It is clear that \cref{thm:main-informal} follows from \cref{thm:main} since $\DTIME(T) \seq \SIZE(O(T \log(T)))$.

\doclearpage
\section{The PCP construction}
\label{sec:PCP-construction}

In this section we formally describe our PCP construction.
In one sentence, the PCP verifier gets a permutation invariant proof
$\NSStrategyRep \colon (\Bits^{\NumWires^2})^\NumRep \to \Bits^\NumRep$,
runs on it linearity test, direct product test, and the parallel repetition of the ALMSS verifier,
and accepts if and only if all tests accepts.

We start by recalling the setting of the PCP verifier of \cite{AroraLMSS98} (the ``linear ALMSS verifier'').
Let $\Language \in \SIZE(\NumWires)$ be a language, and let $\{\Circuit_n\}_{n \in \N}$ be a uniform family of boolean circuits
with $\NumWires = \NumWires(n)$ wires that decides $\Language$. That is, for all inputs $x \in \Bits^n$ of length $n$
it holds that $C_n(x) = 1$ if and only if $x \in \Language$.

For a given length $n$ let $\Circuit \DefineEqual \Circuit_n$ be the circuit corresponding to the computation on inputs of length $n$.
The computation of $\Circuit$ on the input $x$ is viewed as a system
of $\NumConstraints \DefineEqual \NumWires+1$ constraints $\{ \ConstraintPoly_{j}(\WiresVars) = \ConstraintConstant_{j} \}_{j \in [\NumConstraints]}$ over $\NumWires$ boolean variables $\WiresVars = (\Wire_1,\dots,\Wire_\NumWires) \in \Bits^{\NumWires}$, where $\ConstraintPoly_1,\dots,\ConstraintPoly_{\NumConstraints} \colon \Bits^{\NumWires} \to \Bits$ are quadratic polynomials (each involving at most three variables in $\WiresVars$) and $\ConstraintConstant_{1},\dots,\ConstraintConstant_{\NumConstraints}$ are boolean constants.
Each variable represents the value of one of the wires of $\Circuit$ during the computation on the input $x$.
In particular, the first $\InputSize$ variables, $\Wire_1,\dots,\Wire_{\InputSize}$, correspond to the $\InputSize$ input wires,
and the variable $\Wire_\NumWires$ corresponds to the output wire.
The constraints are of three types:
\begin{description}

\item [Input consistency:]
For every $j \in \{1,\dots,\InputSize\}$,
$\ConstraintPoly_{j}(\WiresVars) \DefineEqual \Wire_{i}$ and $\ConstraintConstant_{j} \DefineEqual x_{j}$.

\item [Gate consistency:]\
For every $j \in \{\InputSize+1,\dots,\NumWires\}$,
\begin{itemize}

  \item If the wire represented by the variable $\Wire_{j}$ is an output of an AND gate $\Gate$, where the inputs to $\Gate$ are given by $\Wire_{j_1}, \Wire_{j_2}$, then $\ConstraintPoly_{j}(\WiresVars) \DefineEqual \Wire_{j} - \Wire_{j_1} \cdot \Wire_{j_2}$ and $\ConstraintConstant_{j} \DefineEqual 0$.

  \item If the wire represented by the variable $\Wire_{j}$ is an output of an OR gate $\Gate$, where the inputs to $\Gate$ are given by $\Wire_{j_1}, \Wire_{j_2}$, then $\ConstraintPoly_{j}(\WiresVars) \DefineEqual \Wire_{j} - \Wire_{j_1} - \Wire_{j_2} + \Wire_{j_1} \cdot \Wire_{j_2}$ and $\ConstraintConstant_{j} \DefineEqual 0$.

  \item If the wire represented by the variable $\Wire_{j}$ is an output of a NOT gate $\Gate$, where the input to $\Gate$ is given by $\Wire_{j_1}$, then $\ConstraintPoly_{j}(\WiresVars) \DefineEqual \Wire_{j} - \Wire_{j_1}$ and $\ConstraintConstant_{j} \DefineEqual 1$.

\end{itemize}

\item [Accepting output:]
$\ConstraintPoly_{\NumConstraints}(\WiresVars) \DefineEqual \Wire_{\NumWires}$ and $\ConstraintConstant_{\NumConstraints} \DefineEqual 1$.

\end{description}
We overload notation, and use $\ConstraintPoly_{j}$ to also denote the upper triangular matrix in $\Bits^{\NumWires^2}$ with $\ConstraintPoly_{j}(\WiresVars) = \ip{\ConstraintPoly_{j}, \WiresVars \otimes \WiresVars}$
That is, if $\ConstraintPoly_{j}(\WiresVars) = \sum_{i=1}^{\NumWires} a_i \Wire_i +  \sum_{1\leq i<i' \leq \NumWires} a_{i,i'} \Wire_i \Wire_{i'}$,
then the corresponding matrix has $a_i$ in the diagonal entry $(i,i)$ and $a_{i,i'}$ in the entry $(i,i')$, for $1 \leq i < i' \leq \NumWires$.
Also, for $a \in \Bits^{\NumWires}$, denote by $\Diagonal{a}$ the diagonal matrix in $\Bits^{\NumWires^2}$ whose diagonal is $a$.

\subsection{The linear ALMSS verifier}
\label{sec:pcp-construction-linear-almss}

The \emph{linear} ALMSS verifier of \cite{AroraLMSS98} is defined as follows.

\begin{algorithm}[H]
\caption{The linear ALMSS verifier}\label{alg:linear-ALMSS}

    \InputExplicity{A circuit $\Circuit \colon \Bits^n \to \Bits$ with $\NumWires$ wires, and input $x \in \Bits^n$ to $\Circuit$.}
    \InputOracle{A $\MaxQueries$-non-signaling linear function $\QuasiLinRep \colon \Bits^{\NumWires^2} \to \Bits$.}

    \medskip
    Use the circuit $\Circuit$ and input $x$ to construct the matrices $\ConstraintPoly_1,\dots,\ConstraintPoly_{\NumConstraints} \in \Bits^{\NumWires^2}$ and constants $\ConstraintConstant_{1},\dots,\ConstraintConstant_{\NumConstraints} \in \Bits$ representing the computation of $\Circuit$ on $x$.

    Sample $u,v \in \Bits^{\NumWires}$ and $s \in \Bits^{\NumConstraints}$ uniformly and independently at random.

    Query the oracle $\QuasiLin$ on the 4-element set $S = \{ \Diagonal{u}, \Diagonal{v}, u \otimes v, \sum_{j = 1}^{\NumConstraints} s_j \ConstraintPoly_j \}$.

    \Return ACCEPT if and only if $\QuasiLin(\Diagonal{u})\QuasiLin(\Diagonal{v}) = \QuasiLin(u \otimes v)$ and $\QuasiLin(\sum_{j = 1}^{\NumConstraints} s_j \ConstraintPoly_j) = \sum_{j = 1}^{\NumConstraints} s_j \ConstraintConstant_j$.

\end{algorithm}

\noindent
That is, the verifier makes 4 queries to a linear proof $\QuasiLin \colon \Bits^{\NumWires^2} \to \Bits$ (of exponential length).

\parhead{Completeness}
Completeness of the ALMSS verifier is the same as in the classical setting.
Indeed, $C(x) = 1$, then the classical proof defined by the design is accepted with probability 1.

\parhead{Soundness} For soundness, Chiesa~et~al.~\cite{ChiesaMS19} proved that this construction is indeed sound against linear $O(\log \NumWires)$-non-signaling proofs with soundness error bounded below 1.

\begin{theorem}[Theorem~6 in \cite{ChiesaMS19}]\label{thm:nsLPCP-CMS19}
For any language $\Language \in \SIZE(\NumWires)$ there is an input oblivious PCP system,
where the verifier gets as an explicit input
a circuit $\Circuit$ of size $\NumWires = \NumWires(n)$ deciding $\Language$ and an input $x \in \Bits^n$,
and an oracle access to a \emph{linear} proof $\pi \Bits^{O(\NumWires^2)} \to \Bits$.
The verifier uses $O(\NumWires)$ random coins, makes $4$ queries to the proof that are independent of $x$.
If $x \in \Language$, then there exists a (classical) proof that causes the verifier to accepts with probability $1$.
If $x \not\in \Language$, then for any $O(\log(\NumWires))$-non-signaling linear proof the verifier to accepts with probability at most $39/40$.

That is, we have
\begin{equation*}
\SIZE(\NumWires)
\subseteq
\nsLPCPParams{39/40}{O(\NumWires)}{2^{O(\NumWires^2)}}{4}{O(\log \NumWires)}{O(\NumWires^2)}{O(\InputSize)}
\enspace.
\end{equation*}
\end{theorem}

\subsection{Parallel repetition of the linear ALMSS verifier}
\label{sec:pcp-construction-linear-almss-par-rep}

Next, we consider the $\NumRep$-repeated parallel repetition of the linear ALMSS verifier.
Specifically, the verifier samples $t$ independent sets of queries, makes 4 queries to the PCP over the alphabet $\Bits^\NumRep$,
and accepts if and only if all $t$ sets of answers satisfy the basic linear ALMSS verifier.
Formally, the $\NumRep$-repetition of the linear ALMSS verifier is defined as follows.

\begin{algorithm}[H]
\caption{The $\NumRep$-repeated linear ALMSS verifier}\label{alg:repeated-linear-ALMSS}

    \InputExplicity{A circuit $\Circuit \colon \Bits^n \to \Bits$ with $\NumWires$ wires, and input $x \in \Bits^n$ to $\Circuit$.}
    \InputOracle{A $\NumRep$-repeated $\MaxQueries$-non-signaling linear function $\QuasiLinRep \colon (\Bits^{\NumWires^2})^\NumRep \to \Bits^\NumRep$.}

    \medskip
    Construct the matrices $\ConstraintPoly_1,\dots,\ConstraintPoly_{\NumConstraints} \in \Bits^{\NumWires^2}$
    and constants $\ConstraintConstant_{1},\dots,\ConstraintConstant_{\NumConstraints} \in \Bits$, representing the computation of $\Circuit$ on $x$.

    Sample $u^{(1)},\dots,u^{(\NumRep)},v^{(1)},\dots,v^{(\NumRep)} \in \Bits^{\NumWires}$
    and $s^{(1)},\dots,s^{(t)} \in \Bits^{\NumConstraints}$ independently  and uniformly at random.

    Let $Q_1 = ( \Diagonal{u^{(i)}} )_{i \in [\NumRep]}$;
        $Q_2 = ( \Diagonal{v^{(i)}} )_{i \in [\NumRep]}$;
        $Q_3 = ( u^{(i)} \otimes v^{(i)} )_{i \in [\NumRep]}$;
        $Q_4 = ( \sum_{j = 1}^{\NumConstraints} s^{(i)}_j \ConstraintPoly_j )_{i \in [\NumRep]}$.

    Query the oracle $\QuasiLinRep$ on the 4-element set $S = \{ Q_1,Q_2,Q_3,Q_4\}$.

    Check that $\QuasiLinRep(Q_1)_i \cdot \QuasiLinRep(Q_2)_i = \QuasiLinRep(Q_3)_i \quad \forall i \in [\NumRep]$.

    Check that and $\QuasiLinRep(Q_4)_i = \sum_{j = 1}^{\NumConstraints} s^{(i)}_j \ConstraintConstant_j \quad \forall i \in [\NumRep]$.

    \Return ACCEPT if and only if in the two previous steps all equalities hold.

\end{algorithm}

\noindent
Here, just as in the previous case, the verifier makes 4 queries to a linear proof. However, now the proof is over the alphabet $\Bits^\NumRep$.

\parhead{Completeness}
Completeness of the repeated linear ALMSS verifier is clear.
Indeed, if $C(x) = 1$, then we can take the parallel repetition of the intended classical linear proof,
and it will satisfy the repeated linear ALMSS verifier with probability 1.

\parhead{Soundness}
For soundness we prove in \cref{sec:linear-proof} that if $t \geq O(\log(N))$,
then the verifier is sound against $O(1)$-non-signaling proofs that are \emph{linear} and \emph{consistent}.
The proof works by reducing to the soundness of the \emph{non-repeated} linear ALMSS verifier,
Specifically, we consider a circuit $C$ and an input $x$ to $C$, and consider a $\NumRep$-repeated $\MaxQueries$-non-signaling proof that is accepted with probability at least $\gamma$.
We show that if the proof is linear and consistent, then its flattening is a $\NumRep$-non-signaling (non-repeated) linear proof that satisfies the non-repeated ALMSS verifier with the same probability.
Therefore, if $\gamma \geq 39/40$, then by \cref{thm:nsLPCP-CMS19} we conclude that $C(x) = 1$.

\subsection{From linear PCPs to standard PCPs using linearity and consistency testing}
\label{sec:pcp-construction-full}

So far we have assumed that the given non-signaling proof is linear and consistent.
Below we show how to discard this assumption,
and prove \cref{thm:main} by constructing a PCP verifier that is sound against arbitrary proofs.
This is done by running (the parallel repetition of) the linearity test,
the consistency test, and then feeding (the self-corrected version of) the proof to the linear repeated ALMSS verifier from \cref{alg:repeated-linear-ALMSS}.
We describe the verifier formally below.

\begin{algorithm}[H]
\caption{The $2\NumRep$-repeated ALMSS verifier + consistency test}\label{alg:repeated-ALMSS}

    \InputExplicity{A circuit $\Circuit \colon \Bits^n \to \Bits$ with $\NumWires$ wires, and input $x \in \Bits^n$ to $\Circuit$}
    \InputOracle{A $2\NumRep$-repeated $\MaxQueries$-non-signaling linear function $\NSStrategyRep[2\NumRep] \colon (\Bits^{\NumWires^2})^{2\NumRep} \to \Bits^{2\NumRep}$}

    \medskip
    Sample uniformly random $X,Y \in (\Bits^{\NumWires^2})^{2\NumRep}$.

    Sample uniformly random $W,Z_1,Z_2 \in (\Bits^{\NumWires^2})^{\NumRep}$.

    Sample the four queries $Q_1,Q_2,Q_3,Q_4 \in (\Bits^{\NumWires^2})^{\NumRep}$ of the $\NumRep$-repeated linear ALMSS verifier from \cref{alg:repeated-linear-ALMSS}.
    and let $\DecisionPredicate_{\LIN} \colon (\Bits^{\NumRep})^4 \to \{ACCEPT, REJECT\}$ be the corresponding predicate.

    Define $\Correct\NSStrategyRep \colon (\Bits^{\NumWires^2})^\NumRep \to \Bits^\NumRep$ as in \cref{def:self-correction},
    which makes two queries to $\NSStrategyRep[2\NumRep]$ for every query to $\Correct\NSStrategyRep$.

    Sample an input $S \seq (\Bits^{\NumWires^2})^{2\NumRep}$ to $\NSStrategyRep[2\NumRep]$
    corresponding to querying $\Correct\NSStrategyRep$ on $\{Q_1,Q_2,Q_3,Q_4\}$.

    Query $\NSStrategyRep[2\NumRep]$ on the set $S \cup \{X, Y, Z+Y\} \cup \{[W;Z_1], [W;Z_2]\}$.

    \emph{Linearity test:} Check that $\NSStrategyRep[2\NumRep](X) + \NSStrategyRep[2\NumRep](Y) = \NSStrategyRep[2\NumRep](X+Y)$.

    \emph{Consistency test:} Check that $\NSStrategyRep[2\NumRep]([W;Z_1])_{|W} = \NSStrategyRep[2\NumRep]([W;Z_2])_{|W}$.

    \emph{Linear PCP verifier:} Interpret  $\NSStrategyRep[2\NumRep](S)$ as the answers of $\Correct\NSStrategyRep$ on the query set $(\{Q_1,Q_2,Q_3,Q_4\})$,
    and check that $\Correct\NSStrategyRep(\{Q_1,Q_2,Q_3,Q_4\})$ satisfies $\DecisionPredicate_{\LIN}$.

    \Return ACCEPT if and only if all three steps above accept.
\end{algorithm}

\noindent
That is, the verifier is almost the parallel repetition of the classical ALMSS verifier.
The only difference is that our verifier makes 2 additional queries for the consistency test.

\parhead{Completeness}
Completeness of the repeated ALMSS verifier is clear, as by design the expected proof is linear,
and hence $\NSStrategyRep$ satisfies the linearity constraint with probability 1.
Furthermore, it follows that $\Correct\NSStrategyRep$ is equal to $\NSStrategyRep$, and thus
the predicate $\DecisionPredicate_{\LIN}$ is also satisfied with probability 1.

\parhead{Soundness}
We prove soundness of the PCP system in \cref{alg:repeated-ALMSS} in \cref{sec:proof-of-soundness}.
Specifically, we use \cref{hyp:rounding},and prove that for $t \geq \polylog(N)$ and $\MaxQueries \geq O(1)$,
the verifier is sound against $O(1)$-non-signaling \emph{linear} proofs.
The proof works, again, by reducing to the soundness of the \emph{non-repeated} linear ALMSS verifier.
Specifically, we consider a circuit $C$ and an input $x$ to $C$, and prove that if the PCP verifier accepts
a $\NumRep$-repeated $\MaxQueries$-non-signaling proof $1-\eps$,
then its flattening (or rather the flattening of its self-correction)
is a $(\eps,\NumRep)$-non-signaling (non-repeated) proof $\NSStrategy \colon \Bits^{\NumWires^2} \to \Bits$ that is $(1-\eps)$-linear,
and it satisfies the \emph{non-repeated} linear ALMSS verifier with high probability.
By applying \cref{hyp:rounding} we obtain a $\NumRep^{\Omega(1)}$-non-signaling (non-repeated) proof
that is $(1-\eps')$-linear and satisfies the \emph{non-repeated} linear ALMSS verifier from \cref{alg:linear-ALMSS} with high probability.
By applying a result from \cite{ChiesaMS20} we conclude that
$\NSStrategy$ is close to a linear $\NumRep^{\Omega(1)}$-non-signaling (non-repeated) proof $\Correct\NSStrategy$
that also satisfies the linear ALMSS verifier with high probability,
and thus, by \cref{thm:nsLPCP-CMS19} it follows that $C(x)=1$.
See \cref{sec:proof-of-soundness} for details.

\doclearpage
\section{Proof overview: Soundness}
\label{sec:overview}

In this section we give an overview of the soundness analysis of the parallel repetition of the PCP verifier from \cref{alg:repeated-ALMSS}.
Before describing the actual proof, we first consider soundness against \emph{structured} proofs.
Indeed, this is a common approach in the analysis of PCP systems.
Specifically, we show first that the PCP verifier from \cref{alg:repeated-linear-ALMSS} is sound against such structured proofs.
Then we use local testing and self-correction to show the proofs that
are accepted by the verifier with high probability satisfy the desired properties.

\parhead{Soundness of the linear $\NumRep$-repeated ALMSS verifier}
Fix a circuit $C \colon \Bits^N \to \Bits$ and let $x \in \Bits^N$ be an input to $C$.
Consider a $2\NumRep$-repeated linear ALMSS verifier for $C(x)$, and
suppose that $\QuasiLinRep[2\NumRep] \colon (\Bits^{N^2})^{2\NumRep} \to \Bits^{2\NumRep}$ is a $2\NumRep$-repeated linear and consistent proof such that $2\NumRep$-repeated linear ALMSS verifier the accepts $\QuasiLinRep[2\NumRep]$ with high probability.
According to \cref{claim:flatten-linear-proof} it follows that we can ``flatten'' $\QuasiLinRep[2\NumRep]$ into a $t$-non-signaling \emph{linear} proof
$\QuasiLin \colon \Bits^{N^2}\to \Bits$.
We show that since $\QuasiLinRep[2\NumRep]$ is acepted with high probability by the $2\NumRep$-repeated linear verifier, it follows that the (non-repeated) linear ALMSS verifier accepts $\QuasiLin$ with high probability.
Therefore, by applying the result of \cite{ChiesaMS19} it follows that if $\NumRep > c\log(N)$ for some (sufficiently large) constant $c$, then $C(x)=1$. See Section~\ref{sec:linear-proof} for details.

\parhead{General proofs - forcing consistency using extended linearity test}
Next we prove soundness of the general (i.e., non-linear) parallel repetition PCP against arbitrary $O(1)$-non-signaling proofs.
The general approach is analogous to the approach used for analyzing PCPs, specifically, we first run a test that ``forces'' the proof to be (close to) linear and consistent, and then apply the analysis of the linear proof in the previous paragraph.

More concretely, we fix a circuit $C \colon \Bits^N \to \Bits$ and an input $x \in \Bits^N$ to $C$,
and consider a $2\NumRep$-repeated (non-linear) ALMSS verifier for $C(x)$.
Suppose that $\NSStrategyRep[2\NumRep] \colon (\Bits^{N^2})^{2\NumRep} \to \Bits^{2\NumRep}$ is a $2\NumRep$-repeated proof such that $2\NumRep$-repeated ALMSS verifier accepts $\NSStrategyRep[2\NumRep]$ with high probability.
Our goal is to prove that $C(x) = 1$, and our high level strategy to show it is the following:

\medskip

\begin{enumerate}
    \item First we assume that the proof is permutation invariant as in \cref{def:permutation-folding}.

    \item Suppose the repeated ALMSS verifier accepts $\NSStrategyRep[2\NumRep]$ with high probability $1-\eps$.
    In particular, this implies that $\NSStrategyRep[2\NumRep]$ passes linearity test with at least same probability.

    \item By applying the self-correction procedure, we obtain the self-correction of $\NSStrategyRep[2\NumRep]$.
    The self-correction of $\NSStrategyRep[2\NumRep]$, denoted by $\CorrectNSStrategyRep[\NumRep]$,
    is a $\NumRep$-repeated $\Correct\MaxQueries$-non-signaling proof for $\Correct\MaxQueries = \Omega(\MaxQueries)$
    such that in order to make one query to $\CorrectNSStrategyRep[\NumRep]$ we make $O(1)$ queries to $\NSStrategyRep[2\NumRep]$.
    We prove that $\CorrectNSStrategyRep[\NumRep]$ satisfies the following two properties.

    \begin{enumerate}
      \item $\CorrectNSStrategyRep[\NumRep]$ is $(1-O(\eps))$-linear, i.e., for all $X,Y \in (\Bits^{\NumWires^2})^{\NumRep}$
      it holds that $\Pr[\CorrectNSStrategyRep[\NumRep](X) + \CorrectNSStrategyRep[\NumRep](Y) = \CorrectNSStrategyRep[\NumRep](X+Y)]$.
      That is, the self-correction transforms an \emph{average-case} guarantee about linearity testing $\NSStrategyRep[2\NumRep]$
      into a guarantee that $\CorrectNSStrategyRep[\NumRep]$ satisfies the linearity constraints \emph{for all} $X,Y,X+Y$.
      \item $\CorrectNSStrategyRep[\NumRep]$ is $(1-O(\eps))$-consistent, i.e., for all $Q,Q' \in (\Bits^{\NumWires^2})^{\NumRep}$
      with high probability $\CorrectNSStrategyRep[\NumRep](Q)_j = \CorrectNSStrategyRep[\NumRep](Q')_j$ for all $j \in [\NumRep]$
      such that $Q_j = Q'_j$. Here also, the \emph{average-case} guarantee of the consistency test is
      converted into the \emph{worst-case} guarantee holding for all $Q,Q' \in (\Bits^{\NumWires^2})^{\NumRep}$.
    \end{enumerate}

    \item Next, we let $\Flat\NSStrategy = \flatten{\CorrectNSStrategyRep[\NumRep]}$ be the flattening of $\CorrectNSStrategyRep[\NumRep]$.
    By \cref{claim:flatten-almost-consistent}, $\Flat\NSStrategy \colon \Bits^{\NumWires^2} \to \Bits$ is a almost linear $(O(\eps), \NumRep)$-no-signaling function.
    Furthermore, using the fact that $\Correct\NSStrategyRep$ is $(1-O(\eps))$-linear and is accepted by the repeated ALMSS verifier from \cref{alg:repeated-linear-ALMSS} with high probability,
    we prove that $\Flat\NSStrategy$ is accepted by the (non-repeated) ALMSS verifier from \cref{alg:linear-ALMSS} with high probability.
\end{enumerate}

    At this point we would like to apply \cref{thm:nsLPCP-CMS19} on $\CorrectNSStrategyRep[\NumRep]$,
    and say that since $\Flat\NSStrategy$ is accepted by the AMLSS verifier with high probability, it follows that $C(x) = 1$.
    However, the difficulty in applying~\cref{thm:nsLPCP-CMS19} is that $\Flat\NSStrategy$ is not necessarily non-signaling,
    but only \emph{almost non-signaling} (see \cref{def:almost-ns-function} for reference).
    In order to still apply this result we use \cref{hyp:rounding} to ``round'' $\Flat\NSStrategy$ into a non-signaling proof,
    and then apply \cref{thm:nsLPCP-CMS19} to conclude that $C(x) = 1$. Specifically, we do the following.

\begin{enumerate}
  \setcounter{enumi}{4}

    \item Assuming \cref{hyp:rounding}, there exist a $\NumRep'$-no-signaling function $\NSStrategy$, which is close to $\Flat\NSStrategy$.
    In particular, $\NSStrategy$ is an almost linear non-signaling function.

    \item Using the result of \cite{ChiesaMS19} on linearity testing we get that for some $\linlocalityhat = \Omega(\sqrt{\NumRep'})$ there exists a $\linlocalityhat$-non-signaling linear proof $\QuasiLin$ that is $O(q \eps)$-close to $\NSStrategy'$ on queries sets of size at most $q \leq \linlocalityhat$.

    \item By our choice of parameters, the locality of $\QuasiLin$ is $\linlocalityhat = \Omega(\sqrt{\NumRep'}) > C \log(N)$,
    and by the previous item the linear ALMSS verifier accepts $\QuasiLin$ with high probability.
    Therefore, using \cref{thm:nsLPCP-CMS19} we conclude that $C(x) = 1$.
\end{enumerate}

This completes the overview of the proof. Below we describe each step in detail.

\doclearpage
\section{Soundness of the linear PCP verifier against structured proofs}
\label{sec:linear-proof}

In this section we prove that the $\NumRep$-repeated linear PCP verifier from \cref{alg:repeated-linear-ALMSS}
is sound against linear consistent proofs. Specifically, we prove the following theorem.

\begin{stheorem}\label{thm:consistent-linear-soundness}
Fix a circuit $\Circuit \colon \Bits^n \to \Bits$ with $\NumWires$ wires, and input $x \in \Bits^n$ to $\Circuit$.
Let $\MaxQueries \geq 4$, and $\NumRep$ be positive integers such that $\NumRep \geq K \log(N)$ for some sufficiently large constant $K>0$.
Let $\QuasiLinRep \colon (\Bits^{N^2})^\NumRep \to \Bits^\NumRep$ be a $\MaxQueries$-non-signaling $\NumRep$-repeated linear consistent proof,
and suppose that the $\NumRep$-repeated ALMSS verifier accepts $\QuasiLinRep$ with probability $\geq 39/40$.
Then $C(x) = 1$.
\end{stheorem}

\begin{proof}
Let $\QuasiLinRep \colon (\Bits^{N^2})^\NumRep \to \Bits^\NumRep$ be a $\MaxQueries$-non-signaling $\NumRep$-repeated linear consistent proof,
and suppose that the $\NumRep$-repeated ALMSS verifier accepts $\QuasiLinRep$ with probability $\geq 39/40$.

Let $\Flat\QuasiLin = \flatten{\QuasiLinRep}$ be the flattening of $\QuasiLinRep$ as per \cref{def:flatten}.
Since $\QuasiLinRep$ is linear and consistent, it follows by \cref{claim:flatten-linear-proof} that $\Flat\QuasiLin$ is a $\NumRep$-non-signaling linear function.

Next we show that the non-repeated linear ALMSS verifier accepts $\Flat\QuasiLin$ with probability $>39/40$,
and hence, by \cref{thm:nsLPCP-CMS19} it follows that $\Circuit(x) = 1$.

Indeed, consider the random choices of $u,v \in \Bits^{\NumWires}$ and $s \in \Bits^{\NumConstraints}$ in \cref{alg:linear-ALMSS},
and let $Q^* \in (\Bits^{N^2})^\NumRep$ be a query to $\QuasiLinRep$ that contains the four queries
$Q_{ALMSS} = \{ \Diagonal{u}, \Diagonal{v}, u \otimes v, \sum_{j = 1}^{\NumConstraints} s_j \ConstraintPoly_j \}$ in its first 4 coordinates.
That is, $Q^*_1 = \Diagonal{u}$, $Q^*_2 = \Diagonal{v}$, $Q^*_3 = u \otimes v$, and $Q^*_4 = \sum_{j = 1}^{\NumConstraints} s_j \ConstraintPoly_j$.
Denoting by $\DecisionPredicateLIN$ the predicated in \cref{alg:linear-ALMSS}, by definition of the flattening operation we have
\begin{equation}\label{eq:Pr-accept-flatten-to-first-coord}
  \Pr[\DecisionPredicateLIN(\Flat\QuasiLin(Q_{ALMSS}))]
  =
  \Pr[\DecisionPredicateLIN(\QuasiLinRep(Q^*)_1,\QuasiLinRep(Q^*)_2,\QuasiLinRep(Q^*)_3,\QuasiLinRep(Q^*)_4) = 1]
  \enspace.
\end{equation}
Next, consider the random choices of $u^{(1)},\dots,u^{(\NumRep)},v^{(1)},\dots,v^{(\NumRep)} \in \Bits^{\NumWires}$
and $s^{(1)},\dots,s^{(t)} \in \Bits^{\NumConstraints}$ in \cref{alg:repeated-linear-ALMSS},
and let $Q_1 = ( \Diagonal{u^{(i)}} )_{i \in [\NumRep]}$,
        $Q_2 = ( \Diagonal{v^{(i)}} )_{i \in [\NumRep]}$,
        $Q_3 = ( u^{(i)} \otimes v^{(i)} )_{i \in [\NumRep]}$,
        and $Q_4 = ( \sum_{j = 1}^{\NumConstraints} s^{(i)}_j \ConstraintPoly_j )_{i \in [\NumRep]}$
be the queries made by the repeated ALMSS verifier.

Since $u^{(1)}, v^{(1)}$ and $s^{(1)}$ are distributed identically to the random choices
of $u,v \in \Bits^{\NumWires}$ and $s \in \Bits^{\NumConstraints}$ in \cref{alg:linear-ALMSS},
it follows by consistency of $\QuasiLinRep$ that $\Pr[\DecisionPredicateLIN(\QuasiLinRep(Q^*)_{\mid \{1,2,3,4\}}) = 1]$ is equal to
\begin{equation*}
  \Pr[\DecisionPredicateLIN(
        \QuasiLinRep(Q_1)_1,
        \QuasiLinRep(Q_2)_1,
        \QuasiLinRep(Q_3)_1,
        \QuasiLinRep(Q_4)_1)
         = 1]
  \enspace,
\end{equation*}
i.e., to the probability that $\DecisionPredicateLIN$ accepts
the responses of $\QuasiLinRep$ in the first coordinate of the parallel repetition in \cref{alg:repeated-linear-ALMSS}.
However, since the verifier in \cref{alg:repeated-linear-ALMSS} accepts $\QuasiLinRep$
with probability $\geq 39/40$, it follows in particular, that the first coordinate is accepted with probability $\geq 39/40$,
and hence, by \cref{eq:Pr-accept-flatten-to-first-coord} we conclude that
\begin{equation*}
  \Pr[\DecisionPredicateLIN(\Flat\QuasiLin(Q_{ALMSS}))] \geq 39/40
  \enspace,
\end{equation*}
and hence, by \cref{thm:nsLPCP-CMS19} we have $\Circuit(x) = 1$.
This completes the proof of \cref{thm:consistent-linear-soundness}.
\end{proof}

\doclearpage
\section{Testing and self-correcting repeated non-signaling functions}
\label{sec:test-self-correction}

As shown in \cref{sec:linear-proof}, it is rather straightforward to construct a PCP system
that is sound against repeated non-signaling proofs that are consistent and linear.
Therefore, we would like to make sure that the given proof satisfies these properties.
We ``enforce'' these properties in \cref{alg:repeated-ALMSS} by first running linearity test and consistency test
on a given $\NumRep$-repeated non-signaling proof, and then run the linear PCP on the self-correction of the given proof.
Next, we show that if the tests accept a given proof with high probability,
then its self-correction (almost) satisfies the desired properties,
hence reducing the problem to the structured case.
In this section we analyze the tests and prove guarantees about the self-correction
of any non-signaling function that passes the test with high probability.
Then, in \cref{sec:proof-of-soundness} we use these results on testing and self-correction
in order to analyze the PCP system from \cref{alg:repeated-ALMSS}.

\subsection{Definitions of tests and the self-correction}

\medskip
\parhead{Testing linearity}
Linearity test is a randomized algorithm that given an input function $f$, queries it on 3 inputs and wishes decides whether $f$ is linear or far from linear.
The test was first analyzed in \cite{BlumLR93}. Bellare et al. in \cite{BellareCHKS96} simplified the analysis and proved for any boolean function $f$, the probability that it passes the test is at most $1 - \Delta(f)$,
where $\Delta(f)$ is the normalized Hamming distance of $f$ to the closest linear function.
Extension of \cite{BlumLR93} linearity test to general groups and many other closely related problems have been studied since then
\cite{AumannHRS01, ShpilkaW04, BenOrCLR08, BhattacharyyaKSSZ10, DavidDGEKS17}.
More recently, \cite{ItoV12, Vidick14} and \cite{ChiesaMS20} analyzed the linearity test against quantum strategies and non-signaling strategies.
For our setting, when the functions are of the form
$\NSStrategyRep \colon (\Bits^n)^\NumRep \to \Bits^\NumRep$ the test is as follows.

\begin{definition}[Linearity test~\cite{BlumLR93}]
\label{def:repeated-linearity-test}
Let $\NSStrategyRep \colon (\Bits^n)^\NumRep \to \Bits^\NumRep$ be a $\NumRep$-repeated $\MaxQueries$-non-signaling function.
Linearity test works by uniformly sampling $X, Y  \in (\Bits^n)^\NumRep$,
querying $\NSStrategyRep$ on set $\{X, Y, X+Y\}$, and checking that
$\NSStrategyRep(X) + \NSStrategyRep(Y) = \NSStrategyRep(X+Y)$, i.e., that for all $j \in [\NumRep]$ it holds that
$\NSStrategyRep(X)_j + \NSStrategyRep(Y)_j = \NSStrategyRep(X+Y)_j$.
\end{definition}

In the non-signaling setting, linearity test was analyzed by \cite{ChiesaMS20} for \emph{boolean} functions.
They proved that any $\MaxQueries$-non-signaling boolean function $\NSStrategy$ that passes the linearity test
with probability $1-\eps$ can be self-corrected to a $\floor{\MaxQueries/2}$-non-signaling function $\Correct{\NSStrategy}$
that is $2^{O(\MaxQueries)}\eps$-close to a linear $\floor{\MaxQueries/2}$-non-signaling function $\QuasiLin$.
However, we cannot directly apply their result to our setting, as our functions are not boolean.
Furthermore, adapting the approach of \cite{ChiesaMS20} will give a linear non-signaling function
with the guarantee that the distance between $\Correct{\NSStrategy}$ and a truly linear function $\QuasiLin$
is at most $2^{O(\NumRep \MaxQueries)}\eps$, which is too large for our application.

\medskip
\parhead{Testing consistency}
Next we consider \emph{consistency test}, whose goal is to check that a given $\NumRep$-repeated
non-signaling function is (close to) consistent as per \cref{def:consistent-rep-ns-fun}.
The test works as follows.

\begin{definition}[Consistency test]
\label{def:consistency-test}
Let $\NSStrategyRep[2\NumRep] \colon (\Bits^n)^{2\NumRep} \to \Bits^{2\NumRep}$ be a $2\NumRep$-repeated $\MaxQueries$-non-signaling function for an integer $\NumRep$.
Consistency test chooses $W, Z_1, Z_2  \in (\Bits^n)^{\NumRep}$ uniformly at random,
queries $\NSStrategyRep$ on $\{[W; Z_1], [W; Z_2]\}$,
and checks that $\NSStrategyRep([W; Z_1])_{|W} = \NSStrategyRep([W; Z_2])_{|W}$.
\end{definition}

Similar tests have been studied in the literature in the context of \emph{Direct product testing} in a long series of work \cite{DinurR04,ImpagliazzoKW12, DinurS14-b, DinurN17, GoldenbergC19}.

We prove below that if a $2\NumRep$-repeated $\MaxQueries$-non-signaling
proof $\NSStrategyRep[2\NumRep]$ passes both the linearity test and the consistency test with probability $1-\Soundness$,
then its \emph{self-correction} $\Correct\NSStrategyRep$ is $(1-O(\Soundness))$-linear and $(1-O(\Soundness))$-consistent.
That is, $\Correct\NSStrategyRep$ is close to having the properties we need in order to prove soundness against repeated non-signaling proofs.
Next, we discuss the notion of self-correction,
and prove if $\NSStrategyRep[2\NumRep]$ passes the tests with high probability, then $\Correct\NSStrategyRep$ satisfies the desired properties.

\subsection{Self-correction of a \texorpdfstring{$\NumRep$}{t}-repeated \texorpdfstring{$\MaxQueries$}{k}-non-signaling function}

Below we define the self-correction of a given $\NumRep$-repeated $\MaxQueries$-non-signaling function $\NSStrategyRep$.
Observe that if $\NSStrategyRep$ passes the linearity test with high probability $1-\eps$,
it does not necessarily imply that it satisfies \emph{all} linearity constraints with high probability, i.e.,
it does not imply that $\NSStrategyRep$ is $(1-\eps')$-linear. As a simple example, one may consider the case when
$\NSStrategyRep$ is a deterministic function that is obtained from a linear function by changing some small fraction of its outputs.
The same applies to the consistency test, i.e., satisfying the consistency constraints on average
as opposed to satisfying each consistency constraint with high probability.

A standard approach to transform the ``average-case'' guarantee of the tests into a ``point-wise'' guarantee
is by employing the idea of \emph{self-correction}. Next we define the notion of self-correction suitable for our tests.

\begin{definition}\label{def:self-correction}
Let $\NSStrategyRep[2\NumRep] \colon (\Bits^n)^{2\NumRep} \to \Bits^{2\NumRep}$ be a $2\NumRep$-repeated $\MaxQueries$-non-signaling function.
The \defemph{self-correction} of $\NSStrategyRep[2\NumRep]$,
is a $\NumRep$-repeated $\Correct{\MaxQueries}$-non-signaling function $\CorrectNSStrategyRep[\NumRep] \colon (\Bits^n)^{\NumRep} \to \Bits^{\NumRep}$,
for $\Correct{\MaxQueries} \leq \frac{\MaxQueries}{2}$ defined as follows.

Given a query $Q \in (\Bits^n)^{\NumRep}$, in order to sample $\CorrectNSStrategyRep[\NumRep](Q)$
we uniformly choose $R, W \in (\Bits^n)^{\NumRep}$,
query $\NSStrategyRep[2\NumRep]$ on the set $\{[R; W], [R + Q; W]\}$, and output the first half of
$(\NSStrategyRep[2\NumRep]([R; W]) + \NSStrategyRep[2\NumRep]([R + Q; W])$.

More generally, for a query set $\Correct{S} = \{Q_1,\dots,Q_s\}$ of size $s \leq \Correct{\MaxQueries}$
we  sample $R_i, W_i \in (\Bits^n)^{\NumRep}$ independently, uniformly at random for each $i \in [s]$ ,
query $\NSStrategyRep[2\NumRep]$ on the set
\begin{equation*}
    S = \bigcup_{i=1}^{s} \{[R_i; W_i], [Q_i + R_i; W_i] \}
    \enspace,
\end{equation*}
and output
\begin{equation*}
\CorrectNSStrategyRep[\NumRep](Q_i)_j \DefineEqual (\NSStrategyRep[2\NumRep]([R_i; W_i]) + \NSStrategyRep[2\NumRep]([Q_i + R_i; W_i]))_{j}
\quad \forall j \in [\NumRep]
\enspace.
\end{equation*}
for all $i \in [s]$.
\end{definition}

Observe that the self-correction of $\NSStrategyRep[2\NumRep]$ is indeed a non-signaling function with the appropriate locality parameter.
Indeed, this follows immediately from the assumption that $\NSStrategyRep[2\NumRep]$ is $\MaxQueries$-non-signaling
and the fact that the $R_i, W_i$'s are uniformly random and independent.

\subsection{Self-correction is almost linear and almost consistent}
Next we show that if $\NSStrategyRep[2\NumRep]$ passes both the linearity test and the agreement test with high probability then its self-correction $\Correct\NSStrategyRep$ is almost linear and almost consistent.
Indeed, this average-to-worst-case is a standard step in the analysis of non-signaling PCPs \cite{KalaiRR14,ChiesaMS19}.

\begin{stheorem}
\label{thm:self-correction-structure}
Let $\NSStrategyRep[2\NumRep] \colon (\Bits^n)^{2\NumRep} \to \Bits^{2\NumRep}$ be a $2\NumRep$-repeated $\MaxQueries$-non-signaling function,
and suppose that $\NSStrategyRep[2\NumRep]$ is permutation folded.
If $\NSStrategyRep[2\NumRep]$ passes both the linearity and consistency tests with probability at least $1-\eps$,
then $\CorrectNSStrategyRep[\NumRep]$ is $\Correct{\MaxQueries}$-non-signaling function that is permutation folded,
$(1-4\eps)$-linear, and $(1-8\eps)$-consistent, for for $\Correct{\MaxQueries} = \MaxQueries/2-5$.
\end{stheorem}

\noindent
The rest of this section is devoted to the proof of \cref{thm:self-correction-structure}

\subsubsection{\texorpdfstring{$\CorrectNSStrategyRep[\NumRep]$}{Self-correction} is permutation folded}
We first prove that if $\NSStrategyRep[2\NumRep]$ is permutation folded, then $\CorrectNSStrategyRep[\NumRep]$ is also permutation-folded.
(Recall \cref{def:permutation-folding} for the definition of the permutation folded property and the application of permutations on vectors.)

\begin{lemma}
\label{lemma:self-correction-preserve-folding}
Assuming $\NSStrategyRep[2\NumRep]$ is permutation-folded, $\CorrectNSStrategyRep[\NumRep]$ is also permutation-folded.
\end{lemma}
\begin{proof}
Fix $S = \{Q_1, \dots, Q_\ell\} \seq (\Bits^n)^{\NumRep}$ with $1 \leq \ell \leq \MaxQueries$,
and let $T = \{\pi_1(Q_1), \dots, \pi_\ell(Q_\ell) \}$
for some permutations $\pi_1,\dots\pi_\ell \in S_\NumRep$.
By definition of $\CorrectNSStrategyRep[\NumRep]$ for any $b_1, \dots, b_\ell \in \Bits^\NumRep$ it holds that

\begin{align*}
  \Pr\left[\forall i \in [\ell] \quad \CorrectNSStrategyRep[\NumRep]_S(Q_i) = b_i \right]
  & = \Pr_{R_i,W_i}\left[\forall i \in [\ell] \quad \NSStrategyRep[2\NumRep]([R_i; W_i]) + \NSStrategyRep[2\NumRep]([Q_i + R_i; W_i]) = b_i \right]\\
  & = \Pr_{R_i,W_i}\left[\forall i \in [\ell] \quad \NSStrategyRep[2\NumRep]([\pi(R_i)_i; W_i]) + \NSStrategyRep[2\NumRep]([\pi_i(Q_i + R_i); W_i]) = b_i \right]\\
  & = \Pr\left[\forall i \in [\ell] \quad \CorrectNSStrategyRep[\NumRep]_T(\pi_i(Q_i)) = \pi_i(b_i) \right]
  \enspace,
\end{align*}
as required.
\end{proof}

\subsubsection{\texorpdfstring{$\CorrectNSStrategyRep[\NumRep]$}{Self-correction} is almost linear}
Next, we show that if $\NSStrategyRep$ passes the linearity test with high probability,
then its self-correction is almost linear as per \cref{def:almost-linear-rep-ns-fun}.

\begin{lemma}
\label{lemma:self-correction-almost-linear}
Let $\NSStrategyRep[2\NumRep] \colon (\Bits^n)^{2\NumRep} \to \Bits^{2\NumRep}$ be a $2\NumRep$-repeated $\MaxQueries$-non-signaling function
such that $\MaxQueries \geq 7$.
If $\NSStrategyRep[2\NumRep]$ passes the linearity test with probability at least $1-\eps$, then $\CorrectNSStrategyRep[\NumRep]$ is $(1-4\eps)$-linear.
That is, for any query set $\Correct{S} = \{X, Y, X+Y\} \seq (\Bits^n)^{\NumRep}$ we have $\Pr[\CorrectNSStrategyRep[\NumRep](X) + \CorrectNSStrategyRep[\NumRep](Y) = \CorrectNSStrategyRep[\NumRep](X+Y)] \geq 1 - 4\eps$.
\end{lemma}

The proof is almost the same as in \cite{ChiesaMS20} Theorem 12 ($1 \implies 2$).
The idea is to define a constant number of intermediate events such that each of them holds with high probability by high acceptance probability of the linearity test. Then we put together these intermediate events and derive the desired statement.
\begin{proof}

For $X,Y \in (\Bits^n)^{\NumRep}$ define $Z = X+Y$, and sample $R_X,R_Y,R_Z, W_X,W_Y,W_Z \in (\Bits^n)^{\NumRep}$ uniformly at random independently of each other.
By definition of $\CorrectNSStrategyRep[\NumRep]$ we have
\begin{align*}
\Pr[\CorrectNSStrategyRep[\NumRep](X) + \CorrectNSStrategyRep[\NumRep](Y) = \CorrectNSStrategyRep[\NumRep](X+Y)] \\
\geq
\Pr \bigg[
 &\NSStrategyRep[2\NumRep]([R_X; W_X]) + \NSStrategyRep[2\NumRep]([X+R_X; W_X]) \\
&+
\NSStrategyRep[2\NumRep]([R_Y; W_Y]) + \NSStrategyRep[2\NumRep]([Y+R_Y; W_Y]) \\
&=
\NSStrategyRep[2\NumRep]([R_Z; W_Z]) + \NSStrategyRep[2\NumRep]([Z+R_Z; W_Z])
\bigg]
\enspace.
\end{align*}
Define
\begin{flalign*}
    S_1 \DefineEqual& \{[R_X; W_X], [R_Y; W_Y], [R_Z; W_Z], [X+R_X; W_X], [Y+R_Y; W_Y], [X+Y+R_Z; W_Z] \} \enspace, \\
    S_2 \DefineEqual& \{[R_X; W_X], [R_Z; W_Z], [X+R_X+R_Y; W_X+W_Y], [Y+R_Y; W_Y], [X+Y+R_Z; W_Z] \} \enspace, \\
    S_3 \DefineEqual& \{[R_X; W_X], [X+R_X+R_Y; W_X + W_Y], [Y+R_Y+R_Z; W_Y+W_Z], [X+Y+R_Z; W_Z] \} \enspace, \\
    S_4 \DefineEqual& \{[X+R_X+R_Y; W_X + W_Y], [Y+R_Y+R_Z; W_Y+W_Z], [X+Y+R_X+R_Z; W_X+W_Z] \} \enspace.
\end{flalign*}
Note that $\Cardinality{S_i \cup S_{i+1}} \leq 7 \leq \MaxQueries$ for $i =1,2,3$.

Let $\AddFunction(\cdot)$ be the addition function, and consider the sets $S_1$ and $S_2$.
Then
\begin{align*}
    &\Pr[\AddFunction(\NSStrategyRep[2\NumRep](S_1)) = \AddFunction(\NSStrategyRep[2\NumRep](S_2))] \\
    &= \Pr[\NSStrategyRep[2\NumRep]([X+R_X+R_Y; W_X + W_Y] + \NSStrategyRep[2\NumRep]([X+R_X; W_X]) = \NSStrategyRep[2\NumRep]([R_Y; W_Y])]\enspace.
\end{align*}
Observing that the distribution on the right hand side is exactly as in the linearity test,
we get that
\begin{equation*}
    \Pr[\AddFunction(\NSStrategyRep[2\NumRep](S_1)) = \AddFunction(\NSStrategyRep[2\NumRep](S_2))] \geq 1 - \eps \enspace.
\end{equation*}
Similarly, we have
\begin{align*}
    &\Pr[\AddFunction(\NSStrategyRep[2\NumRep](S_2)) = \AddFunction(\NSStrategyRep[2\NumRep](S_3))] \\
    &= \Pr[\NSStrategyRep[2\NumRep]([R_Z; W_Z] + \NSStrategyRep[2\NumRep]([Y+R_Y; W_Y]) = \NSStrategyRep[2\NumRep]([Y+R_Y+R_Z; W_Y+W_Z])] \\
    &\geq 1-\eps \enspace,
\end{align*}
and
\begin{align*}
    &\Pr[\AddFunction(\NSStrategyRep[2\NumRep](S_3)) = \AddFunction(\NSStrategyRep[2\NumRep](S_4))] \\
    &= \Pr[\NSStrategyRep[2\NumRep]([R_X; W_X] + \NSStrategyRep[2\NumRep]([X+R_Z; W_Z]) = \NSStrategyRep[2\NumRep]([X+Y+R_X+R_Z; W_X+W_Z])] \\
    &\geq 1-\eps \enspace.
\end{align*}
Therefore,
\begin{align*}
\abs{\Pr[\AddFunction(\NSStrategyRep[2\NumRep](S_1)) = 0] - \Pr[\AddFunction(\NSStrategyRep[2\NumRep](S_4)) = 0]  }
&\leq \sum_{i = 1}^3 \abs{ \Pr[\AddFunction(\NSStrategyRep[2\NumRep](S_i)) = 0] - \Pr[\AddFunction(\NSStrategyRep[2\NumRep](S_{i+1})) = 0] }\\
&\leq 3 \eps
\enspace.
\end{align*}
Finally, note that
\begin{equation*}
    \Pr[\AddFunction(\NSStrategyRep[2\NumRep](S_4)) = 0] \geq 1-\eps
    \enspace,
\end{equation*}
because the distribution of $S_4$ is equal to the distribution of a three tuple used for linearity testing.
Therefore,
\begin{equation*}
    \Pr[\CorrectNSStrategyRep[\NumRep](X) + \CorrectNSStrategyRep[\NumRep](Y) = \CorrectNSStrategyRep[\NumRep](X+Y)]
    \geq
    \Pr[\AddFunction(\NSStrategyRep[2\NumRep](S_1)) = 0]
    \geq 1 - 4\eps
    \enspace,
\end{equation*}
as required.
\end{proof}

\subsubsection{\texorpdfstring{$\CorrectNSStrategyRep[\NumRep]$}{Self-correction} is almost consistent}
Finally, we prove in \cref{lemma:self-correction-almost-consistent}
that if $\NSStrategyRep$ passes the consistency test with high probability,
then its self-correction is almost consistent.
Before proving it we need the following claim.

\begin{claim}
\label{claim:zero-on-zeros}
Let $\NSStrategyRep[2\NumRep] \colon (\Bits^n)^{2\NumRep} \to \Bits^{2\NumRep}$ be a $2\NumRep$-repeated $\MaxQueries$-non-signaling function
such that $\MaxQueries \geq 6$.
Suppose that $\NSStrategyRep[2\NumRep]$ passes both linearity and consistency tests with probability at least $1 - \eps$.
Then for any $Q \in (\Bits^n)^{\NumRep}$ it holds that
\begin{equation*}
    \Pr\left[ \CorrectNSStrategyRep[\NumRep](Q)_j = 0 \quad \forall j \in [\NumRep] \mbox{ such that } Q_j = 0^n \right] > 1 - 4\eps
\end{equation*}
\end{claim}

\begin{proof}
The key observation here is that for a uniformly random $R,W \in (\Bits^n)^{\NumRep}$ it holds that
\begin{equation}\label{eq:consistency-using-lin}
    \Pr\left[\NSStrategyRep[2\NumRep]([Q+R; W])_j = \NSStrategyRep[2\NumRep]([R; W]))_j \quad \forall j = \NumRep+1, \dots, 2\NumRep \right] \geq 1 - 4\eps
    \enspace.
\end{equation}
(Note that \cref{eq:consistency-using-lin} does not follow from consistency testing since $R$ and $Q+R$ are not independent.)
Indeed, let $R',R'',W' \in (\Bits^n)^{\NumRep}$ be sampled uniformly at random, independently of all other random variables.
Then, since $\NSStrategyRep[2\NumRep]$ passes linearity test with probability at least $1-\eps$, it follows that
with probability at least $1 - 2\eps$ the following equalities hold:
\begin{align*}
        \NSStrategyRep[2\NumRep]([Q+R; W]) =& \NSStrategyRep[2\NumRep]([Q+R''; W']) + \NSStrategyRep[2\NumRep]([R+R''; W+W']) \\
        \NSStrategyRep[2\NumRep]([R; W]) =& \NSStrategyRep[2\NumRep]([R'; W']) + \NSStrategyRep[2\NumRep]([R+R'; W+W'])
\end{align*}
If these two equalities hold, then
\begin{align*}
    \NSStrategyRep[2\NumRep]([Q+R; W]) + \NSStrategyRep[2\NumRep]([R; W])
    =& \NSStrategyRep[2\NumRep]([Q+R''; W']) + \NSStrategyRep[2\NumRep]([R'; W']) \\
    &+ \NSStrategyRep[2\NumRep]([R+R''; W+W']) + \NSStrategyRep[2\NumRep]([R+R'; W+W'])
    \enspace.
\end{align*}
Noting that the queries $\{[Q+R''; W'], [R'; W']\}$ are distributed as in the consistency test,
it follows that
\begin{equation}\label{eq:consistency-apply1}
    \Pr\left[\NSStrategyRep[2\NumRep]([Q+R''; W'])_j = \NSStrategyRep[2\NumRep]([R'; W']))_j \quad \forall j = \NumRep+1, \dots, 2\NumRep \right] \geq 1 - \eps
    \enspace.
\end{equation}
By the same argument we have
\begin{equation}\label{eq:consistency-apply2}
    \Pr\left[\NSStrategyRep[2\NumRep]([R+R''; W+W'])_j = \NSStrategyRep[2\NumRep]([R+R';W+W']))_j \quad \forall j = \NumRep+1, \dots, 2\NumRep \right] \geq 1 - \eps
    \enspace.
\end{equation}
These immediately imply \cref{eq:consistency-using-lin}.

In order to complete the proof let $\pi \in S_{2\NumRep}$ be an arbitrary permutation such that for all $j \in [t]$, $\pi(j) \in \{\NumRep+1, \dots, 2\NumRep\}$.
Then,
\begin{align*}
    &\Pr\left[ \CorrectNSStrategyRep[\NumRep](Q)_j = 0 \quad \forall j \in [\NumRep] \mbox{ such that } Q_j = 0^n \right] \\
    &=
    \Pr\left[\NSStrategyRep[2\NumRep]([Q+R; W])_j = \NSStrategyRep[2\NumRep]([R; W]))_j \quad \forall j \in [\NumRep] \mbox{ such that } Q_j = 0^n \right] \\
    &=
    \Pr\left[\NSStrategyRep[2\NumRep](\pi([Q+R; W]))_{\pi(j)} = \NSStrategyRep[2\NumRep](\pi([R; W])))_{\pi(j)} \quad \forall j \in [\NumRep] \mbox{ such that } Q_j = 0^n \right] \\
    &\geq 1 - 4\eps
    \enspace,
\end{align*}
where the last inequality follows from \cref{eq:consistency-using-lin} together with the permutation invariance of $\NSStrategyRep[2\NumRep]$.
\end{proof}

The following lemma, saying that $\CorrectNSStrategyRep[\NumRep]$ is $(1-O(\eps))$-consistent, follows almost immediately from \cref{claim:zero-on-zeros}.

\begin{lemma}\label{lemma:self-correction-almost-consistent}
Let $\NSStrategyRep \colon (\Bits^n)^{2\NumRep} \to \Bits^{2\NumRep}$ be a $2\NumRep$-repeated $\MaxQueries$-non-signaling function
such that $\MaxQueries \geq 7$, and suppose that $\NSStrategyRep[2\NumRep]$ is permutation folded.

If $\NSStrategyRep[2\NumRep]$ passes both linearity and consistency tests with probability $1-\eps$,
then $\CorrectNSStrategyRep[\NumRep]$ is $(1-8\eps)$-consistent.
That is, for any two queries $X, Y \in (\Bits^n)^{\NumRep}$ to $\CorrectNSStrategyRep[\NumRep]$ it holds that
\begin{equation*}
    \Pr \left[ \CorrectNSStrategyRep[\NumRep](X)_j = \CorrectNSStrategyRep[\NumRep](Y)_j
        \quad \forall j \in [\NumRep] \mbox{ such that } X_j = Y_j  \right] \geq 1 - 8\eps
    \enspace.
\end{equation*}
\end{lemma}

\begin{proof}
    Let $J = \{j \in [\NumRep]: X_j = Y_j\}$.
    Consider the query set $S = \{X,Y,Z = X+Y\}$,
    and note that $Z_j = 0^n$ for all $j \in J$.
    Therefore, by \cref{claim:zero-on-zeros} it follows that
    $\Pr\left[ \CorrectNSStrategyRep[\NumRep](Z)_j = 0 \quad \forall j \in J \right] > 1 - 4\eps$.
    By applying \cref{lemma:self-correction-almost-linear} we have
    $\Pr[\CorrectNSStrategyRep[\NumRep](X) + \CorrectNSStrategyRep[\NumRep](Y) = \CorrectNSStrategyRep[\NumRep](Z)] \geq 1 - 4\eps$.
    Therefore, by the union bound we conclude that
    $\Pr \left[ \CorrectNSStrategyRep[\NumRep](X)_j = \CorrectNSStrategyRep[\NumRep](Y)_j \quad \forall j \in J \right] \geq 1 - 8\eps$,
    thus concluding the proof of \cref{lemma:self-correction-almost-consistent}.
\end{proof}

\cref{thm:self-correction-structure} is an immediate conclusion from
\cref{lemma:self-correction-preserve-folding},
\cref{lemma:self-correction-almost-linear},
and \cref{lemma:self-correction-almost-consistent}.

\doclearpage
\section{Proof of \texorpdfstring{\cref{thm:main}}{the main theorem }}
\label{sec:proof-of-soundness}

Below we prove \cref{thm:main}. Specifically, we show that assuming \cref{hyp:rounding}
the PCP construction in \cref{alg:repeated-ALMSS} is sound against $O(1)$-non-signaling proofs.
\cref{thm:main} follows immediately from the following statement.

\begin{theorem}\label{thm:main-soundness}
Fix a circuit $\Circuit \colon \Bits^n \to \Bits$ with $\NumWires$ wires, and input $x \in \Bits^n$ to $\Circuit$.
Let $\MaxQueries \geq 18$ be a sufficiently large positive constant, and $\NumRep$ be a positive integer such that $\NumRep \geq K \log^{2/\expHypothesis}(\NumWires)$ for some sufficiently large constant $K>0$.
Let $\NSStrategyRep[2\NumRep] \colon (\Bits^{\NumWires^2})^{2\NumRep} \to \Bits^{2\NumRep}$ be a $\MaxQueries$-non-signaling $2\NumRep$-repeated linear consistent proof,
and suppose that $\NSStrategyRep[2\NumRep]$ is permutation invariant.
If $2\NumRep$-repeated ALMSS verifier from \cref{alg:repeated-ALMSS} accepts $\NSStrategyRep[2\NumRep]$ with probability $\geq 1-\eps$
for some sufficiently small $\eps$, then $C(x) = 1$.
\end{theorem}

The proof follows the steps outlined in \cref{sec:overview}.

\begin{proof}
Fix a $2\NumRep$-repeated $\MaxQueries$-non-signaling proof $\NSStrategyRep[2\NumRep]$ that satisfies the repeated ALMSS verifier from \cref{alg:repeated-ALMSS} with probability at least $1-\eps$.
In particular, $\NSStrategyRep[2\NumRep]$ passes the linearity test and the consistency test with probability at least $1 - \eps$.

By applying \cref{thm:self-correction-structure} we conclude that $\CorrectNSStrategyRep[\NumRep]$, the self-correction of $\NSStrategyRep[2\NumRep]$,
is a 4-no-signaling function that is $(1-4\eps)$-linear and $(1-8\eps)$-consistent.
Furthermore, $\CorrectNSStrategyRep[\NumRep]$ satisfies the linear $\NumRep$-repeated verifier from \cref{alg:repeated-linear-ALMSS} with probability at least $1-\eps$.

Define $\Flat\NSStrategy = \flatten{\CorrectNSStrategyRep[\NumRep]}$
to be the flattening of $\CorrectNSStrategyRep[\NumRep]$, as per \cref{def:flatten}.
Then, by \cref{claim:flatten-linear-proof}
the function $\Flat\NSStrategy \colon \Bits^{\NumWires^2} \to \Bits$ is $(1-(4+3\cdot8)\eps)$-linear $(8\eps, \NumRep)$-non-signaling.
Furthermore, since $\CorrectNSStrategyRep[\NumRep]$ is $(1-8\eps)$-consistent,
and satisfies the linear repeated verifier from \cref{alg:repeated-linear-ALMSS} with probability at least $1-\eps$,
it follows that $\Flat\NSStrategy$ satisfies the (non-repeated) linear verifier
from \cref{alg:linear-ALMSS} with probability at least $1-9\eps$.

Next, we use \cref{hyp:rounding} to round $\Flat \NSStrategy$ to an exactly non-signaling function $\NSStrategy$ close to it.
Specifically,
since $\Flat\NSStrategy$ is $(1-28\eps)$-linear $(8\eps, \NumRep)$-non-signaling,
by \cref{hyp:rounding} there exist $\NumRep'$-non-signaling function
$\NSStrategy \colon \Bits^{\NumWires^2} \to \Bits$ for $\NumRep' \geq \NumRep^{\expHypothesis} = K' \log^{2}(\NumWires)$,
such that $\Delta_{4}(\Flat \NSStrategy, \NSStrategy) \leq \eps'$, where $\eps' = \eps'_{\hyp}(28\eps)$.
In particular, since $\Flat \NSStrategy$ is $(1-28\eps)$-linear, it follows that $\NSStrategy$ is $(1-28\eps-\eps')$-linear,
and satisfies the PCP verifier from \cref{alg:linear-ALMSS} with probability at least $1-9\eps-\eps'$.

Next, we apply the following theorem on almost linear non-signaling functions from~\cite{ChiesaMS19}.
The theorem says that any almost linear function $\NSStrategy$ can be ``rounded'' into an exactly non-signaling function $\QuasiLin$,
such that the two are close to each other on predicates that depend on a small number of coordinates.

\begin{theorem}[Theorem~7 in \cite{ChiesaMS19}]\label{thm:almost-lin-function}
Let $\NumRep', \linlocalityhat\in \N$ and $\eps \in (0,1/400]$ be such that $\NumRep' = \Omega(\frac{\linlocalityhat}{\eps} \cdot ( \linlocalityhat + \log \frac{1}{\eps} ))$. Suppose that $\NSStrategy \colon \Bits^n \to \Bits$ is a $\NumRep'$-non-signaling function such that for all $x,y \in \Bits^n$ it holds that $\Pr[\NSStrategy(x) + \NSStrategy(y) = \NSStrategy(x+y)] \geq 1 - \eps$. Then there exists a linear $\linlocalityhat$-non-signaling function $\QuasiLin  \colon \Bits^n \to \Bits$ such that for all query sets $Q \seq \Bits^\linlength$ of size $\abs{Q} \leq \linlocalityhat$ and for all events $\Event \seq \Bits^{Q}$ it holds that
\begin{equation*}
    \abs{ \Pr[\NSStrategy(Q) \in \Event]  - \Pr[\QuasiLin(Q) \in \Event] }
    \leq
    (6 \abs{Q} + 3) \sqrt{\eps}
\enspace.
\end{equation*}
\end{theorem}
\begin{remark}
    Actually, Theorem~7 in \cite{ChiesaMS19} assumes that linearity test accepts $\NSStrategy$ with high probability,
    and the conclusion of the theorem holds for its self-correction $\Correct\NSStrategy$.
    However, if we make the stronger assumption that $\Pr[\NSStrategy(x) + \NSStrategy(y) = \NSStrategy(x+y)] \geq 1 - \eps$
    holds for all $x,y \in \Bits^n$, then by following the proof, it is easy to see that the conclusion holds for $\NSStrategy$, without the self-correction.
\end{remark}

By applying \cref{thm:almost-lin-function} on $\NSStrategy$, and using it for all 4-ary predicates used by \cref{alg:linear-ALMSS},
it follows that there exists a linear $\linlocalityhat$-non-signaling function $\QuasiLin \colon \Bits^n \to \Bits$
that satisfies the PCP verifier from \cref{alg:linear-ALMSS} with probability at least $1-\Correct\eps$
for $\Correct\eps = 1-9\eps-\eps'- (6\cdot4+3)\sqrt{9\eps+\eps'} = 1-O(\sqrt{\eps+\eps'})$.
In particular, if $\NumRep' > K' \log^2(\NumWires)$ for a sufficiently large constant $K'$, then $\linlocalityhat \geq \bar{C}\log(\NumWires)$.
Therefore, if $\eps>0$ is a sufficiently small constant, it follows that
the PCP verifier from \cref{alg:linear-ALMSS} accepts $\QuasiLin$ with probability greater than $39/40$,
and by \cref{thm:nsLPCP-CMS19} we conclude that $\Circuit(x) = 1$.
This completes the proof of \cref{thm:main-soundness}.
\end{proof}

\doclearpage
\section{Conclusions and open problems}
\label{sec:open-problems}

In this paper we establish a conditional result on the existence of a PCP system that is sound against non-signaling proofs with constant locality.
There are several natural research directions left open for future work.

\parhead{Resolving the hypothesis}
The implications of \cref{hyp:rounding} motivates the study of geometry of non-signaling proofs.
In particular, as a natural intermediate step toward settling \cref{hyp:rounding},
one can study the validity of a weaker version of hypothesis, requiring that the rounded proof
is close to the given almost non-signaling proof on all subsets of size at most 2 (instead of 4) assuming that $\NSStrategy$ is linear (instead of almost linear),
i.e., requiring that $\Distance{2}(\NSStrategy, \NSStrategy') \leq \eps'$.
We remark that although \cref{hyp:rounding} requires that $\Distance{4}(\NSStrategy, \NSStrategy')$ is small,
in fact, it suffices to show that $\Distance{3}(\NSStrategy, \NSStrategy')$ is small,
i.e., prove the hypothesis for subsets of size at most 3.

\parhead{Reducing the alphabet}
While we answer \cref{qst:ns-pcp-thm-exp} affirmatively up to \cref{hyp:rounding}, we may require the proof to be of smaller alphabet.
In the classical PCPs literature, the standard technique for alphabet reduction is known as \emph{proof composition},
where the given ``outer'' proof over large alphabet is composed with a collection of ``inner proofs of proximity''
over small alphabet~\cite{BenSassonGHSV06,DinurR04}. Indeed, this component plays an important role in
the modular proof of the PCP theorem. It would be interesting to apply a similar approach to the non-signaling setting.

\parhead{Extending our approach to polynomial size nsPCPs}
Our PCP construction is based on \emph{exponential-length} PCP construction of \cite{AroraLMSS98},
which encodes proofs using the Hadamard code of exponential length.
The effective proof length or alternatively the number of random bits used by the verifier are very important parameters for downstream applications.
In order to reduce the proof length, it is natural to replace the linear encoding with low-degree polynomial encoding~\cite{BabaiFL91,BabaiFLS91}.
Indeed, \cite{KalaiRR13, KalaiRR14} proved that such an approach gives a PCP system that is sound against non-signaling proofs,
albeit with locality $\polylog(T)$.
It would be interesting to see if the parallel repetition of their verifier is sound against non-signaling proofs with constant locality.

\doclearpage
\bibliographystyle{alpha}
\bibliography{references}

\appendix

\doclearpage
\section{Discussion on \texorpdfstring{\cref{hyp:rounding}}{our hypothesis}}
\label{sec:hypothesis-discussion}

As mentioned in the introduction, we can think of $\MaxQueries$-non-signaling functions as points in the polytope
$P_\MaxQueries \subseteq \mathbb{R}^d$, for $d = \sum_{i=0}^{\InputSize} {\InputSize \choose i} 2^i$,
which corresponds the solutions of the $\MaxQueries$'th level relaxation of the Sherali-Adams hierarchy.
Analogously, we can think of $(\eps, \MaxQueries)$-non-signaling functions as points in the polytope
$P_\MaxQueries^{\eps} \subseteq \mathbb{R}^d$,
which corresponds the solutions of the noisy version of the $\MaxQueries$'th level relaxation of Sherali-Adams hierarchy,
where for any two sets $S,T \seq [\InputSize]$ the marginal distributions induced by $P_S$ on $S \cap T$
is $\eps$-close in total variation distance to the marginal distributions induced by $P_T$ on $S \cap T$.
Then, \cref{hyp:rounding} can be rephrased as follows: for any $p \in P_{\MaxQueries}^{\eps}$ there exists
$p' \in P_{\MaxQueries'}$ such that $\Distance{4}(p, p') \leq \eps'$.

We remark that sensitivity analysis of linear programs has been studied in the past (see, e.g., \cite{Schrijver86} Section 10).
However, the parameters obtained by these results seem to be too weak for our application.
Nonetheless, it is possible that this approach could still work for our setting, since we are looking for an approximate solution with respect to the $\Distance{4}$ distance, which is rather non-standard.

\medskip

In~\cite{ChiesaMS20}, the following lemma, in the same spirit as the hypothesis, was proved.

\begin{lemma}[(see {\cite[Lemma C.3]{ChiesaMS20}})]
\label{lemma:rounding-cms}
For every $(\eps,\MaxQueries)$-non-signalling function $\NSStrategy \colon \Domain \to \Bits$ there exists $\MaxQueries$-non-signalling function $\NSStrategy'$ such that $\Distance{\MaxQueries}(\NSStrategy, \NSStrategy') \leq O(4^\MaxQueries \cdot \eps)$
\end{lemma}

While the guarantee of $O(4^\MaxQueries \cdot \eps)$ on the distance in the lemma is too large for our applications,
\cref{hyp:rounding} is somewhat more specific, and it is plausible that proving it is easier than improving \cref{lemma:rounding-cms}.
We discuss \cref{hyp:rounding} below.

\begin{enumerate}
  \item Note that unlike \cref{lemma:rounding-cms}, \cref{hyp:rounding} assumes that $\NSStrategy$ is almost linear.
    We do not know whether this is essential, however, it is reasonable to believe that being almost linear
    adds constraints on the structure of $\NSStrategy$, thus making it easier to prove \cref{hyp:rounding}.

  \item In \cref{hyp:rounding} the requirement on the distance between the given almost non-signaling function,
    and the rounded function is only on sets of size at most 4. In fact, it is not difficult to see that proving that
    $\Distance{3}(\NSStrategy, \NSStrategy') \leq \eps'$ also suffices for the applications.
    This seems to be a significant relaxation compared to $\Distance{\MaxQueries}$ proved in \cref{lemma:rounding-cms}.


  \item In fact, our proof of soundness would go through even with a weaker version of the hypothesis,
  where we replaced the ``worst-case'' notion of $\Distance{4}$ with the ``average-case''.
  Specifically, given an $(\eps, \MaxQueries)$-almost non-signaling proof $\NSStrategy$ that satisfies
  \emph{every constraint} of the linear ALMSS verifier with high probability,
  we want the rounded proof to satisfy the linear ALMSS verifier
  with high probability with respect to the distribution induced by the verifier
  on the 4-query sets.

  Furthermore, since our almost non-signaling proof $\NSStrategy$ is obtained
  by flattening the repeated proof $\CorrectNSStrategyRep[\NumRep]$,
  we may assume that $\NSStrategy$ satisfies
  every constraints of the $\Omega(\MaxQueries)$-sequential repetition of the linearity test,
  i.e., for some $\ell = \Omega(\MaxQueries)$ it holds that
  \begin{equation*}
  \forall x_1, y_1, \dots, x_{\ell}, y_{\ell} \in \Bits^\InputSize
  \quad
  \Pr\left[\NSStrategy(x_i) + \NSStrategy(y_i) = \NSStrategy(x_i+y_i)  \quad \forall i \in [\ell] \right] \geq 1 - \eps
  \enspace,
  \end{equation*}
  and, similarly, $\NSStrategy$ satisfies \emph{every constraint} of the $\Omega(\NumRep)$-sequential repetition
  of the linear ALMSS verifier with high probability,
  and the goal is to get a rounded proof to satisfy the linear ALMSS verifier
  with high probability with respect to the distribution induced by the verifier
  on the 4-query sets.

  \item An alternative way to prove our main theorem is to prove that \cref{thm:nsLPCP-CMS19}
  holds for almost non-signaling proofs. This question seems to be well motivated by the application to \emph{delegation of computation}.
  Indeed, Kalai et.~al~\cite{KalaiRR14} constructed PCP systems (of polynomial size)
  that are sound against $(\eps, \polylog(\NumWires))$-non-signaling proofs,
  for some negligible $\eps>0$. However, their proof seems to break for constant $\eps>0$.
  Our work motivates studying the power of almost non-signaling proofs for constant $\eps>0$.
\end{enumerate}

\end{document}